\documentclass{amsart}[11pt]
\usepackage{amsmath, amsfonts, amsthm, amssymb,mathtools}
\usepackage{subfigure}
\usepackage{stmaryrd}
\usepackage{verbatim}
\usepackage{hyperref}
\usepackage{color}
\usepackage{xcolor}
\usepackage{ulem}
\usepackage{graphicx}
\usepackage{enumerate}
\usepackage{bbm}
\numberwithin{equation}{section}
\newtheorem{theorem}{Theorem}[section]
\newtheorem{corollary}[theorem]{Corollary}
\newtheorem{lemma}[theorem]{Lemma}
\newtheorem{proposition}[theorem]{Proposition}
\newtheorem{conjecture}[theorem]{Conjecture}

\theoremstyle{definition}
\newtheorem{definition}[theorem]{Definition}
\newtheorem{remark}[theorem]{Remark}
\newtheorem{assumption}[theorem]{Assumption}

\newtheorem{example}[theorem]{Example}

\makeatletter
\newcommand*{\rom}[1]{\expandafter\@slowromancap\romannumeral #1@}
\makeatother

\newcommand{\ind}{1\hspace{-2.1mm}{1}} 
\newcommand{\I}{\mathtt{i}}
\newcommand{\eps}{\varepsilon}
\newcommand{\BS}{\mathrm{BS}}
\newcommand{\EE}{\mathbb{E}}
\newcommand{\VV}{\mathbb{V}}
\newcommand{\NN}{\mathbb{N}}
\newcommand{\RR}{\mathbb{R}}
\newcommand{\QQ}{\mathbb{Q}}

\newcommand{\PP}{\mathbb{P}}
\newcommand{\RV}{\mathcal{R}}
\newcommand{\Crm}{\mathrm{C}}
\newcommand{\Urm}{\mathrm{U}}
\newcommand{\m}{\mathfrak{m}}
\newcommand{\n}{\mathfrak{b}}
\newcommand{\A}{\mathfrak{a}}
\newcommand{\B}{\mathfrak{b}}
\newcommand{\vm}{\mathfrak{v}_-}
\newcommand{\vp}{\mathfrak{v}_+}
\newcommand{\M}{\mathrm{M}}
\newcommand{\D}{\mathrm{d}}

\newcommand{\E}{\mathrm{e}}
\newcommand{\Drm}{\mathrm{D}}

\newcommand{\HH}{\mathrm{H}}
\newcommand{\LDP}{\mathrm{LDP}}
\newcommand{\LL}{\mathrm{L}}

\newcommand{\sgn}{\mathrm{sgn}}
\newcommand{\Nn}{\mathcal{N}}
\newcommand{\Vv}{\mathcal{V}}

\newcommand{\Oo}{\mathcal{O}}
\newcommand{\Cc}{\mathcal{C}}
\newcommand{\Ee}{\mathcal{E}}
\newcommand{\Ss}{\mathcal{S}}
\newcommand{\Dd}{\mathcal{D}}

\newcommand{\Pp}{\mathcal{P}}
\newcommand{\usm}{\overline{u}_{-}}
\newcommand{\usp}{\overline{u}_{+}}
\newcommand{\usmp}{\overline{u}_{\pm}}

\newcommand{\ccm}{\underline{\mathfrak{c}}}
\newcommand{\ccp}{\overline{\mathfrak{c}}}
\newcommand{\LLm}{\underline{\Lambda}^{*}}

\newcommand{\Lf}{\mathfrak{L}}

\newcommand{\Hh}{\mathcal{H}}
\newcommand{\rrho}{\overline{\rho}}
\newcommand{\cf}{\mathfrak{c}}
\newcommand{\ttheta}{\overline{\theta}}

\newcommand{\Ff}{\mathcal{F}}

\newcommand{\ggm}{\underline{\gamma}}
\newcommand{\ggp}{\overline{\gamma}}

\newcommand{\PS}{(\Omega, \Ff, (\Ff_t)_{t\geq 0}, \PP)}
\newcommand{\lambtilde}{\widetilde{\Lambda}}
\newcommand{\dtilde}{\widetilde{\Drm}}
\newcommand{\ctilde}{\widetilde{\Crm}}

\newcommand{\that}{\breve{t}}

\def\equalDistrib{\,{\buildrel \text{(Law)} \over =}\,}

\graphicspath{{Figures/}  }

\title{The randomised Heston model}
\date{\today}
\author{Antoine Jacquier, Fangwei Shi}
\address{Department of Mathematics, Imperial College London}
\email{a.jacquier@imperial.ac.uk, fangwei.shi12@imperial.ac.uk}

\thanks{The authors would like to thank Stefan Gerhold and Archil Gulisashvili for useful discussions.
AJ acknowledges financial support from the EPSRC First Grant EP/M008436/1.
FS is funded by a mini-DTC scholarship from the Department of Mathematics, Imperial College London.
Numerical implementations are carried out on the collaborative platform Zanadu (www.zanadu.io).}
\date{\today}
\begin{document}
\maketitle

\begin{abstract}
We propose a randomised version of the Heston model--a widely used stochastic volatility model in mathematical finance--assuming that 
the starting point of the variance process is a random variable.
In such a system, we study the small-and large-time behaviours of the implied volatility,
and show that the proposed randomisation generates a short-maturity smile much steeper (`with explosion')
than in the standard Heston model, 
thereby palliating the deficiency of classical stochastic volatility models in short time.
We precisely quantify the speed of explosion of the smile for short maturities
in terms of the right tail of the initial distribution, and in particular show that an explosion rate 
of~$t^\gamma$ ($\gamma\in[0,1/2]$) for the squared implied volatility--as observed on market data--can be obtained by a suitable choice of randomisation.
The proofs are based on large deviations techniques and the theory of regular variations.
\end{abstract}

\section{Introduction}\label{sec:intro}
Implied volatility is one of the most important observed data in financial markets and
represents the price of European options,
reflecting market participants' views.
Over the past two decades, a number of (stochastic) models have been proposed in order
to understand its dynamics and reproduce its features.
In recent years, a lot of research has been devoted to understanding the asymptotic behaviour 
(large strikes~\cite{benaim08,benaim09,caravenna15}, small / large maturities~\cite{FordeJac09, FordeJac11, forde12, tankov11})
of the implied volatility in a large class of models in extreme cases;
these results not only provide closed-form expressions (usually unavailable) for the implied volatility, 
but also shed light on the role of each model parameter and, ultimately on the efficiency of each model.

Continuous stochastic volatility models driven by Brownian motion
effectively fit the volatility smile (at least for indices);
the widely used Heston model, for example, is able to fit the volatility surface for 
almost all maturities~\cite[Section~3]{gatheral06}, 
but becomes inaccurate for small maturities.
The fundamental reason is that small-maturity data is much steeper 
(for small strikes)--the so-called `short-time explosion'--than 
the smile generated by these stochastic volatility models
(a detailed account of this phenomenon can be found in the volatility bible~\cite[Chapters 3 and 5]{gatheral06}).
To palliate this issue, Gatheral (among others) comments that jumps should be added in the stock dynamics;
the literature on the influence of the jumps is vast, and we only mention 
here the clear review by Tankov~\cite{tankov11} in the case of exponential L\'evy models,
where the short-time implied volatility explodes at a rate of~$|t\log t|$ for small~$t$.
To observe non-trivial convergence (or divergence), 
Mijatovi\'c and Tankov~\cite{mijatovic12} introduced maturity-dependent strikes,
and studied the behaviour of the smile in this regime.

As an alternative to jumps, a portion of the mathematical finance community has recently been advocating
the use of fractional Brownian motion (with Hurst parameter~$H<1/2$) 
as driver of the volatility process.
Al\`os, Le\'on and Vives~\cite{Alos} first showed that such a model is indeed
capable of generating steep volatility smiles for small maturities (see also the recent work by Fukasawa~\cite{Fukasawa}), 
and Gatheral, Jaisson and Rosenbaum~\cite{gatheral14} recently showed that 
financial data exhibits strong evidence that volatility is ‘rough’ 
(an estimate for SPX volatility actually gives $H \approx 0.14$).
Guennoun, Jacquier and Roome~\cite{guennoun14} investigated a fractional version of the Heston model, 
and proved that as~$t$ tends to zero
the squared implied volatility explodes at a rate~$t^{H-1/2}$.
This is currently a very active research area, and the reader is invited to consult~\cite{BayerRough, Euch1, Euch2, Euch3, forde16} 
for further developments.
This is however not the end of the story--yet--as computational costs for simulation
are a severe concern in fractional models.

We propose here a new class of models, namely standard stochastic volatility models 
(driven by standard Brownian motion) where the initial value of the variance is randomised,
and focus our attention to the Heston version.
The motivation for this approach originates from the analysis of forward-start smiles 
by Jacquier and Roome~\cite{jacquier13,JR15}, who proved that 
the forward implied volatility explodes at a rate of~$t^{1/4}$ as~$t$ tends to zero.
A simple version of our current study is the `CEV-randomised Black-Scholes model' introduced in~\cite{jacquier15},
where the Black-Scholes volatility is randomised according to the distribution generated from an independent CEV process;
in this work, the authors proved that this simplified model generates the desired explosion of the smile.
The Black-Scholes randomised setting where the volatility has a discrete distribution
corresponds to the lognormal mixture dynamics studied in~\cite{brigoMixt, brigo02}.
We push the analysis further here;
our intuition behind this new type of models is that the starting point of the volatility process 
is actually not observed accurately, but only to some degree of uncertainty.
Traders, for example, might take it as the smallest (maturity-wise) observed at-the-money implied volatility.
Our initial randomisation aims at capturing this uncertainty.
This approach was recently taken by Mechkov~\cite{mechkov15}, 
considering the ergodic distribution of the CIR process as starting distribution,
who argues that randomising the starting point captures potential hidden variables.
One could also potentially look at this from the point of view of uncertain models,
and we refer the reader to~\cite{Fouque} for an interesting related study.
The main result of our paper is to provide a precise link between the explosion rate of the implied volatility smile
for short maturities and the choice of the (right tail of the) initial distribution of the variance process.
The following table (a more complete version with more examples can be found in Table~\ref{table:nonlin})
gives an idea of the range of explosion rates that can be achieved through our procedure;
for each suggested distribution of the initial variance, we indicate the asymptotic behaviour 
(up to a constant multiplier) of the (square of the) out-of-the-money implied volatility smile
(in the first row, the function~$f$ will be determined precisely later, but the absence of time-dependence 
is synonymous with absence of explosion).
\begin{table}[ht]
\centering
\begin{tabular}{c c c c c c}
\hline\hline
Name & Behaviours of~$\sigma^2_t(x)$ ($x\ne 0$) & Reference\\ [0.5ex] 
\hline
Uniform & $f(x)$ & Equation~\ref{eq:IVStandard}\\
Exponential($\lambda$) & $|x|t^{-1/2}$ & Theorem~\ref{thm:fat-tailImpvolAsymps}\\
$\chi$-squared & $|x|t^{-1/2}$ & Theorem~\ref{thm:fat-tailImpvolAsymps}\\
Rayleigh & $x^{2/3} t^{-1/3}$ & Theorem~\ref{thm:ThinTail}\\
Weibull ($k>1$) & $(x^2/t)^{1/(1+k)}$ & Theorem~\ref{thm:ThinTail}\\ [1ex]
\hline
\end{tabular}
\label{table:intro}
\end{table}

The rest of the paper is structured as follows:
we introduce the randomised Heston model in Section~\ref{sec:model description}, 
and discuss its main properties.
Section~\ref{sec:appetiser} is a numerical appetiser to give a flavour of the quality of such a randomisation.
Section~\ref{sec:Asymptotics} is the main part of the paper,
in which we prove large deviations principles for the log-price process,
and translate them into short- and large-time behaviours of the implied volatility.
In particular, we prove the claimed relation between the explosion rate of the small-time smile
and the tail behaviour of the initial distribution.
The small-time limit of the at-the-money implied volatility is, as usual in this literature, 
treated separately in Section~\ref{sec:ATM}.
Section~\ref{sec:dynamic} includes a dynamic pricing framework:
based on the distribution at time zero and the evolution of the variance process,
we discuss how to re-price (or hedge) the option during the life of the contract.
Finally, Section~\ref{sec:unbounded} presents examples of common initial distributions,
and numerical examples.
The appendix gathers some reminders on large deviations and regular variations,
as well as proofs of the main theorems.


\textbf{Notations:}
Throughout this paper, 
we denote~$\sigma_t(x)$ the implied volatility of a European Call or Put option 
with strike~$\E^x$ and time to maturity~$t$.
For a set~$\Ss$ in a given topological space we denote by~$\mathcal{S}^o$ and~$\overline{\Ss}$ 
its interior and closure. 
Let~$\RR_+ := [0,\infty)$,~$\RR_+^* := (0,\infty)$, and~$\RR^* := \RR\setminus\{0\}$. 
For two functions~$f$ and~$g$, and $x_0\in\RR$, we write~$f\sim g$ as~$x$ tends to $x_0$
if $\lim\limits_{x\rightarrow x_0} f(x)/g(x) =1$.
If a function~$f$ is defined and locally bounded on~$[x_0,\infty)$, and $\lim\limits_{x\uparrow \infty}f(x) = \infty$,
define $f^\leftarrow(x):=\inf\left\{y\geq [x_0,\infty): f(y)>x\right\}$ 
as its generalised inverse. 
Also define the sign function as $\sgn(u) := \ind_{\{u\geq 0\}} - \ind_{\{u< 0\}}$.
Finally, for a sequence $(Z_t)_{t\geq 0}$ satisfying a large deviations principle as~$t$ tends to zero
with speed~$g(t)$ and good rate function~$\Lambda_Z^*$ 
(Appendix~\ref{app:LDPLDP}) we use the notation
$Z\sim\LDP_0(g(t), \Lambda_Z^*)$.
If the large deviations principle holds as~$t$ tends to infinity, we denote it by $\LDP_\infty(\cdots)$.

\section{Model and main properties}\label{sec:model description}
On a filtered probability space~$\PS$ supporting two independent Brownian motions~$W^{(1)}$ 
and~$W^{(2)}$, we consider a market with no interest rates, and propose the following dynamics for 
the log-price process:
\begin{equation}\label{eq:RandomHestonSDE}
\begin{array}{rll}
\D X_{t} & = \displaystyle
 -\frac{1}{2}V_{t} \D t + \sqrt{V_{t}}\left(\rho\,\D W^{(1)}_t + \rrho\,\D W^{(2)}_t\right), 
 & X_{0} = 0,\\
\D V_{t} & = \displaystyle \kappa (\theta-V_{t})\D t + \xi \sqrt{V_{t}} \D W^{(1)}_{t}, 
& V_{0} \equalDistrib \Vv,
\end{array}
\end{equation}
where~$\rho\in[-1,1]$, $\rrho := \sqrt{1-\rho^{2}}$, 
and~$\kappa, \theta, \xi$ are strictly positive real numbers.
Here~$\Vv$ is a continuous random variable, independent of the filtration $(\Ff_t)_{t\geq 0}$,
for which the interior of the support is of the form $(\vm, \vp)$ 
for some $0\leq \vm \leq \vp \leq \infty$, 
with moment generating function $\M_{\Vv}(u) := \EE\left(\E^{u \Vv}\right)$, 
for all $u\in\Dd_\Vv:=\{u\in\RR: \EE\left(\E^{u \Vv}\right)<\infty\}\supset (-\infty, 0]$,
and we further assume that~$\Dd_{\Vv}$ contains at least an open neighbourhood of the origin,
namely that $\m:=\sup\left\{u\in\RR: \M_{\Vv}(u)<\infty\right\}$ belongs to~$(0,\infty]$.
Then clearly all positive moments of~$\Vv$ exist.
Existence and uniqueness of a solution to this stochastic system is guaranteed as soon as~$\Vv$ 
admits a second moment~\cite[Chapter 5, Theorem 2.9]{KS91}.
Notice that the process $(X, V)$ is not adapted to the filtration $(\Ff_t)_{t\geq 0}$ due to the lack of information on~$\Vv$ in~$\Ff_t$.
The process is Markovian, however, with respect to the augmented filtration 
$\sigma(\Ff_t\vee \sigma(\Vv))_{t\geq 0}$.

When~$\Vv$ is a Dirac distribution ($\vm = \vp$), 
the system~\eqref{eq:RandomHestonSDE} corresponds to the standard Heston model~\cite{Heston}, 
and it is well known that the stock price process~$\exp(X)$ is a~$\PP$-martingale;
it is trivial to check that it is still the case for~\eqref{eq:RandomHestonSDE}.
Behaviour~\cite{zeliade10}, asymptotics~\cite{FordeJac11, forde12, FordeJacMij10}, estimation and calibration~\cite{andersen07, zeliade10} of the Heston model have been treated at length in several papers, and we refer the interested reader to this literature for more details about it;
we shall therefore always assume that $\vm < \vp$.
\begin{remark}\label{rmk:props-of-v}
For any $t\geq 0$, the tower property for conditional expectation yields
\begin{equation*}
\begin{array}{ll}
\EE(V_t) 
&= \EE[\EE(V_t|\Vv)]
= \theta\left(1-\E^{-\kappa t}\right) + \E^{-\kappa t}\EE(\Vv),\\
\VV(V_t) 
&= \EE[\VV(V_t|\Vv)] + \VV[\EE(V_t|\Vv)]
= \displaystyle \E^{-2\kappa t}\left(\VV(\Vv) + \frac{\xi^2}{\kappa}\left(\E^{\kappa t}-1\right)\EE(\Vv)\right)+ \frac{\xi^2\theta}{2\kappa}\left(1-\E^{-\kappa t}\right)^2.
\end{array}
\end{equation*}
Consider the standard Heston model ($\vm=\vp =: V_0$), 
and construct~$\Vv$ such that $\EE(\Vv) = V_0$.
Then, for any time $t\geq 0$, both random variables $V_t$ (in~\eqref{eq:RandomHestonSDE} 
and in the standard Heston model) have the same expectation;
however, the randomisation of the initial variance increases the variance by $\E^{-2\kappa t}\VV(\Vv)$.
As time tends to infinity, it is straightforward to show that the randomisation 
preserves the ergodicity of the variance process, with a Gamma distribution as invariant measure, 
with identical mean and variance:
$$
\lim_{t\uparrow \infty} \EE(V_t)
= \theta
\qquad\text{and}\qquad
\lim_{t\uparrow \infty}\VV(V_t)
 = \frac{\xi^2\theta}{2\kappa}.
$$
\end{remark}

For any $t\geq 0$, let $\M(t,u)$ denote the moment generating function (mgf) of~$X_{t}$:
\begin{equation}\label{eq:MGFX}
\M(t,u):= \EE\left(\E^{u X_{t}}\right),
\qquad\text{for all }u\in\Dd_\M^t := \left\{u\in\RR: \EE\left(\E^{u X_{t}}\right)<\infty\right\}.
\end{equation}
The tower property yields directly
\begin{equation}\label{eq:TowerProperty}
\M(t,u) = \EE \left(\E^{uX_t} \right)
= \EE\left(\EE\left(\E^{uX_t}|\Vv\right)\right)
= \EE\left(\E^{\Crm(t, u) + \Drm(t, u)\Vv}\right)
= \E^{\Crm(t, u)} \M_{\Vv}\left(\Drm(t, u)\right),
\end{equation}
where the functions~$\Crm$ and~$\Drm$ arise directly from the (affine) representation 
of the moment generating function of the standard Heston model, 
recalled in Appendix~\eqref{eq:HestonMGF}.

\section{Practical appetiser and relation to model uncertainty}\label{sec:appetiser}
\subsection{The bounded support case: a practical appetiser}
Before diving into the technical statements and proofs of asymptotic results in Section~\ref{sec:Asymptotics}, 
let us provide a numerical hors-d'oeuvre, teasing the appetite of the reader regarding the practical relevance of the randomisation.
As mentioned in the introduction, the main drawback of classical continuous-path stochastic volatility models 
(without randomisation and driven by standard Brownian motions),
is that the small-maturity smile they generate is not steep enough to reflect the reality of the market.
Graph~\ref{graph:uniformIntro} below represents a comparison of the implied volatility surface generated by the standard Heston model with
$$
\kappa = 2.1,\qquad
\theta= 0.05,\qquad
V_0 = 0.06,\qquad
\rho = -0.6,\qquad
\xi = 0.1,
$$
and that of the Heston model randomised by a uniform distribution with
$\vm = 0.04$ and $\vp = 0.082$.
From the trader's point of view, this could be understood as uncertainty on the actual value of~$V_0$
(see also~\cite{Fouque} for a related approach).
Clearly, the randomisation steepens the smile for small maturities,
while its effect fades away as maturity becomes large.
\begin{figure}[h!]
\includegraphics[scale=0.4]{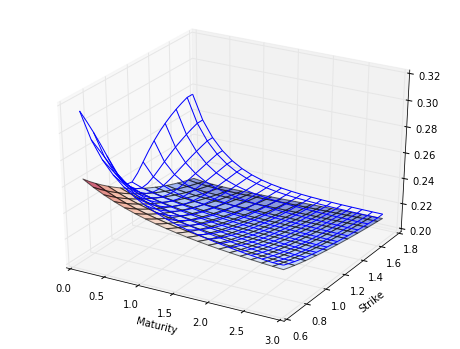}
\caption{Volatility surfaces of standard Heston (coloured) and with randomisation 
$\Vv\equalDistrib \mathcal{U}(4\%,8.2\%)$.}
\label{graph:uniformIntro}
\label{fig:appetiser}
\end{figure}
This numerical example intuitively yields the following informal conjecture:
\begin{conjecture}
Under randomisation of the initial volatility, the smile `explodes' for small maturities.
\end{conjecture}

We shall provide a precise formulation--and exact statements--of this conjecture.
Despite the appearances in Figure~\ref{fig:appetiser},
the conjecture is actually  false when the initial distribution has bounded support,
such as in the uniform case here.
However, as will be detailed in Section~\ref{ex:UniformDist}, greater steepness of the smile 
(compared to the standard Heston model) does appear for a wide range of strikes,
but not in the far tails 
(this is quantified precisely, 
as well as the at-the-money curvature in the uncorrelated case, in Section~\ref{ex:UniformDist}).
This leads us to believe that, even if `explosion' does not actually occur in the bounded support case,
this assumption may still be of practical relevance given the range of traded strikes.

\section{Asymptotic behaviour of the randomised model}\label{sec:Asymptotics}
This section is the core of the paper, and relates the explosion of the implied volatility smile in small times
to the tail behaviour of the randomised initial variance.
Section~\ref{sec:preliminary} (Proposition~\ref{prop:limit_of_cgf}) provides the short-time behaviour of 
the cumulant generating function (cgf) of
the random sequence~$(X_t)_{t\geq 0}$, 
and relates it to the choice of the initial distribution~$\Vv$.
This paves the way for a large deviations principle for the sequence~$(X_t)_{t\geq 0}$.
Section~\ref{sec:ThinTail} concentrates on the case where both~$\m$ and~$\vp$ are infinite:
Theorem~\ref{thm:ThinTail} indicates that 
the squared implied volatility has an explosion rate of~$t^\gamma$ with~$\gamma\in(0,1/2)$.
The case where~$\m<\vp=+\infty$ is covered in Section~\ref{sec:FatTail}, where
an explosion rate of~$\sqrt{t}$ is obtained.
Section~\ref{sec:LargeTime} provides the large-time asymptotic behaviour
of the implied volatility in our randomised setting;
in particular, the long-term similarities between standard and randomised Heston models are present in this section.
Finally, Section~\ref{sec:ATM} covers the singular case of the small-time at-the-money implied volatility.

\subsection{Preliminaries}\label{sec:preliminary}
As a first step in understanding the behaviour of the implied volatility, 
we analyse the short-time limit of the rescaled cgf of the sequence~$(X_t)_{t\geq 0}$.
To do so, let~$h: \RR_+ \to \RR_+$ be a smooth function, which can be extended at zero by continuity
with $h(0) := \lim_{t\downarrow 0}h(t)=0$. 
In light of~\eqref{eq:TowerProperty}, for any $t\geq 0$, we introduce the effective domain of the moment generating function 
of the rescaled random variable~$X_t/h(t)$:
$$
\Dd_{t}:=\left\{u\in\RR: \M\left(t,\frac{u}{h(t)}\right) < \infty\right\},
$$
as well as the following sets, for any $t>0$:
$$
\Dd^t_{\Vv}:=\left\{u\in\RR:\M_{\Vv}\circ\Drm\left(t, \frac{u}{h(t)}\right)<\infty\right\},
\qquad
\Dd^*:=\liminf_{t\downarrow 0}\Dd_{t} = \bigcup_{t>0}\bigcap_{s\leq t}\Dd_s,
\qquad
\Dd_{\Vv}^* := \liminf_{t\downarrow 0}\Dd^t_{\Vv} = \bigcup_{t>0}\bigcap_{s\leq t}\Dd^s_{\Vv}.
$$
We now denote the pointwise limit 	
$\Lambda_{h}(u):=\lim\limits_{t\downarrow 0}\Lambda_{h}\left(t,u/h(t)\right)$, where 
\begin{equation}\label{eq:LambdaH}
\Lambda_{h}\left(t, \frac{u}{h(t)}\right)
 := h(t)\log \M\left(t,\frac{u}{h(t)}\right).
\end{equation}
The seemingly identical notations for the function and its pointwise limit should not create any confusion
in this paper.
Introduce further the real numbers $u_-\leq 0$ and $u_+\geq 1$ and the function $\Lambda:(u_-, u_+)\to\RR$:
\begin{equation}\label{eq:lambda_H}
\begin{array}{rl}
& \left\{
\begin{array}{rl}
u_{-} &:= \displaystyle \frac{2}{\xi\rrho}\arctan\left(\frac{\rrho}{\rho}\right)\ind_{\{\rho<0\}}
- \frac{\pi}{\xi}\ind_{\{\rho=0\}}
+ \frac{2}{\xi\rrho}\left(\arctan\left(\frac{\rrho}{\rho}\right)-\pi\right)\ind_{\{\rho>0\}},\\
u_{+} &:= \displaystyle \frac{2}{\xi\rrho}\left(\arctan\left(\frac{\rrho}{\rho}\right)+\pi\right)\ind_{\{\rho<0\}}
+\frac{\pi}{\xi}\ind_{\{\rho=0\}}
+ \frac{2}{\xi\rrho}\arctan\left(\frac{\rrho}{\rho}\right)\ind_{\{\rho>0\}},
\end{array}
\right.
\\
& \displaystyle \Lambda(u) := \frac{u}{\xi(\rrho\mathrm{cot}\left(\xi\rrho u/2\right)-\rho)}.
\end{array}
\end{equation}
The following proposition, whose proof is postponed to Appendix~\ref{subsec:proof_thm_limit_of_cgf}, 
summarises the limiting behaviour of~$\Lambda_h(\cdot,\cdot)$ as~$t$ tends to zero. 
In view of Remark~\ref{rem:h}(ii) below, we shall only consider power functions
of the type $h(t) \equiv c t^\gamma$.
It is clear that there is no loss of generality by taking~$c=1$, as it only acts as a space-scaling factor.
We shall therefore replace the notation~$\Lambda_h$ by~$\Lambda_\gamma$
to highlight the power exponent in action. 
\begin{proposition}\label{prop:limit_of_cgf}
Let $h(t) = t^{\gamma}$, with $\gamma \in (0,1]$.
As~$t$ tends to zero, the following pointwise limit holds:
\begin{equation*}
\Lambda_\gamma(u)
 := \lim_{t\downarrow 0}\Lambda_{\gamma}\left(t, \frac{u}{t^{\gamma}}\right) = 
\left\{
\begin{array}{llll}
0, & u\in\RR, & \text{if } \gamma \in (0,1/2), &\text{for any }\Vv,\\
0, & u\in\RR, & \text{if } \gamma \in [1/2,1), & \vp<\infty,\\
\Lambda(u)\vp, & u\in (u_-, u_+), 
& \text{if }\gamma =1,
& \vp<\infty,\\
\LL_\pm\ind_{\{u=\pm\sqrt{2\m}\}}, 
& u\in [-\sqrt{2\m}, \sqrt{2\m}], 
& \text{if } \gamma = 1/2,
& \vp=\infty, \m<\infty,
\end{array}
\right.
\end{equation*}
and is infinite elsewhere, where $\LL_{\pm} \in [0,\infty]$.
Whenever~$\gamma>1$ (for any~$\Vv$), or~$\m<\infty$ and~$\gamma>1/2$, the limit is infinite everywhere except at the origin.
\end{proposition}
We shall call the (pointwise) limit `degenerate' whenever it is either equal to zero everywhere or 
zero at the origin and infinity everywhere else.
In Proposition~\ref{prop:limit_of_cgf}, only the last two cases are not degenerate.
\begin{remark}\label{rem:h}\ 
\begin{enumerate}[(i)]
\item The case where~$\vp$ and~$\m$ are both infinite is treated separately, in Section~\ref{sec:ThinTail},
as more assumptions are needed on the behaviour of the distribution of~$\Vv$.
\item If~$h$ is not a power function, the proofs of Proposition~\ref{prop:limit_of_cgf} and Theorem~\ref{lem:ThinTail} indicate that we only need to compare the order of~$h$ with orders of~$t^{1/2}$ and~$t$. 
Any non-power function then yields degenerate limits.
\item In the last case, 
$\LL_{\pm}$ depend on the explicit form of the mgf of~$\Vv$.
Example~\ref{ex:NcCS} illustrates this.
\end{enumerate}
\end{remark}

When the random initial distribution~$\Vv$ has bounded support ($\vp<\infty$), 
Proposition~\ref{prop:limit_of_cgf} indicates that the only possible speed factor is $\gamma=1$,
and a direct application 
of the G\"artner-Ellis theorem (Theorem~\ref{thm:Gartner})
implies a large deviations for the sequence~$(X_t)_{t\geq 0}$;
adapting directly the methodology from~\cite{FordeJac09}, 
we obtain the small-time behaviour of the implied volatility:

\begin{corollary}\label{Coro:BddSupport}
If~$\vp<\infty$, 
then~$X\sim\LDP_0(t, \Lambda_{\vp}^*)$ with
$\Lambda_{\vp}^*(x):= \sup\left\{ux - \Lambda(u)\vp: u\in(u_-,u_+)\right\}$ and
\begin{equation}\label{eq:IVStandard}
\lim_{t\downarrow 0}\sigma_t^2(x) = \frac{x^2}{2\Lambda_{\vp}^*(x)},
\qquad\text{for all }x \ne 0.
\end{equation}
\end{corollary}
Approximations, in particular around the at-the-money $x=0$, of the rate function~$\Lambda_{\vp}^*$,
and hence of the small-time implied volatility, can also be found in~\cite[Theorem 3.2]{FordeJac09}, 
and apply here directly as well.
Further, as discussed in detail in Section~\ref{ex:UniformDist}, 
higher order terms in the small-time expansion of $\sigma_t^2(x)$ can be obtained
if the mgf of the initial randomisation is known in closed form.

\subsection{The thin-tail case}\label{sec:ThinTail}
In the case $\m=\infty$, Proposition~\ref{prop:limit_of_cgf} is not sufficient as several different behaviours 
can occur.
In this case, which we naturally coin `thin-tail', a more refined analysis is needed, 
and the following assumption shall be of uttermost importance:
\begin{assumption}[Thin-tail]\label{Assu:V0}
$\vp = \infty$ and~$\Vv$ admits a smooth density~$f$ with
$\log f(v)\sim -l_1 v^{l_2}$ as~$v$ tends to infinity, for some~$(l_1,l_2)\in \RR_+^*\times(1,\infty)$.
\end{assumption}

For notational convenience, we introduce the following two special rates of convergence
$\frac{1}{2}<\ggm<1<\ggp$, 
and two positive constants~$\ccm$,~$\ccp$:
\begin{equation}\label{eq:SpecialRates}
\ggm := \frac{l_2}{1+l_2},
\qquad\qquad
\ggp := \frac{l_2}{l_2-1},
\qquad\qquad
\ccm := (2l_1l_2)^{\frac{1}{1+l_2}},
\qquad\qquad
\ccp := (2l_1l_2)^{\frac{1}{1-l_2}}.
\end{equation}

The following theorem is the main result of this thin-tail section, and provides
both a large deviations principle for the log-stock price process
as well as its implications on the small-maturity behaviour of the implied volatility.
Define the function $\LLm: \RR\to\RR_+$ by
\begin{equation}\label{eq:I}
\LLm(x) := \frac{\ccm}{2\ggm} x^{2\ggm},\qquad\text{for any }x\text{ in }\RR.
\end{equation}

\begin{theorem}\label{thm:ThinTail}
Under Assumption~\ref{Assu:V0}, $X\sim\LDP_0(t^{\ggm},\LLm)$ with~$\LLm$ given in~\eqref{eq:I}, and, for any $x\ne 0$, 
$$
\lim_{t\downarrow 0}t^{1-\ggm}\sigma_t^2(x)
 = \ccm^{-1}\ggm x^{2(1-\ggm)}.
$$
\end{theorem}

In exponential L\'evy models, the implied variance~$\sigma_t^2(x)$ for non-zero~$x$ 
explodes at a rate~$|t\log t|$~\cite[Proposition~4]{tankov11}. 
Theorem~\ref{thm:ThinTail} implies that in a thin-tail randomised Heston model we have a much slower explosion rate of~$t^{\eta}$ with~$\eta \in (0,1/2)$. 
In~\cite{mijatovic12} the authors commented that market data suggests that implied volatility with decreasing maturity still has a reasonable range of values and does not explode significantly,
which might provide empirical grounds justifying the potential value of this randomised model as an alternative to the exponential L\'evy models.
The theorem relies on the study of the asymptotic behaviour of the rescaled mgf of~$X_t$:
\begin{lemma}\label{lem:ThinTail}
Under Assumption~\ref{Assu:V0}, the only non-degenerate speed factor is 
$\gamma = \ggm$, and 
\begin{equation}\label{eq:LambdaHThin}
\Lambda_{\ggm}(u) 
= \frac{\ccp}{2\ggp} u^{2\ggp},
\quad\text{for any }u\text{ in }\RR.
\end{equation}
\end{lemma}

Assumption~\ref{Assu:V0} in particular implies that the function $\log f$ is regularly varying with index~$l_2$
(which we denote $|\log f| \in \RV_{l_2}$, see also Appendix~\ref{app:LDPRegVar} for a review of 
and useful results on regular variation).
Without this slightly stronger assumption, however, the constant in~\eqref{eq:LambdaHThin}--essential
to compute precisely the rate function governing the corresponding large deviations principle
(Theorem~\ref{thm:ThinTail})--would not be available.
In order to prove the lemma and hence the theorem, 
let us first state and prove the following result:
\begin{lemma}\label{lemma: log-mgf-asymptotics}
If $|\log f|\in \RV_l$ ($l>1$), then
$
\log\M_{\Vv}(z) \sim \left(l-1\right) \left(\frac{z}{l}\right)^{\frac{l}{l-1}}\psi(z)$
at infinity,
with~$\psi\in\RV_0$ defined as 
$$
\psi(z) := \left(\frac{z}{|\log f|^\leftarrow(z)}\right)^\leftarrow z^{\frac{l}{1-l}}.
$$
\end{lemma}
\begin{proof}
Since $|\log f| \in \RV_l$, 
Bingham's Lemma (Lemma~\ref{lem:Bingham}) implies 
$\log\PP(\Vv\geq x) = \log\int_{x}^{\infty}\E^{\log f(y)}\D y \sim \log f(x)$,
as~$x$ tends to infinity,
and the result follows from Kasahara's Tauberian theorem~\cite[Theorem~4.12.7]{bingham89}.
\end{proof}

\begin{proof}[Proof of Lemma~\ref{lem:ThinTail} and of Theorem~\ref{thm:ThinTail}]
By Lemma~\ref{lemma:suffi_condition}, the mgf of~$\Vv$ is well-defined on~$\RR_+$.
Lemmas~\ref{lemma: log-mgf-asymptotics} and~\ref{lemma:Dasymp} imply that as~$t$ tends to zero,
\begin{equation*}
t^\gamma\log\M_\Vv\left(\Drm\left(t,\frac{u}{t^\gamma}\right)\right)
=
\left\{
\begin{array}{rll}
\displaystyle t^\gamma\log\M_{\Vv}\left(\frac{u^2}{2}t^{1-2\gamma}
\left(1 + \Oo\left(t^{1-\gamma}\right)\right)\right)
& \displaystyle \sim
\frac{\ccp}{2\ggp}u^{2\ggp}
t^{\ggp(1 - \gamma/\ggm)},
 & \text{when }\gamma \in (1/2,1),\\
\displaystyle t\log\M_{\Vv}\left(\frac{\Lambda(u)}{t}\left(1 + \Oo(t)\right)\right)
& \displaystyle \sim 
\frac{\ccp}{2\ggp}2^{\ggp}\Lambda(u)^{\ggp}t^{1-\ggp},
 & \text{when }\gamma=1.
\end{array}
\right.
\end{equation*}
For $u\ne 0$ the right-hand side is well defined with non-zero limit if and only 
if~$\gamma =\ggm\in(1/2,1)$;
the case~$\gamma=1$ does not yield any non-degenerate behaviour, and the lemma follows.

The large deviations principle stated in Theorem~\ref{thm:ThinTail} 
is a direct consequence of Lemma~\ref{lem:ThinTail} and the G\"artner-Ellis theorem 
(Theorem~\ref{thm:Gartner}), 
noting that the function~$\Lambda_{\ggm}$ in~\eqref{eq:LambdaHThin} satisfies all the required conditions and admits~$\LLm$ as Fenchel-Legendre transform.
The translation of this asymptotic behaviour into implied volatility follows the same lines as in~\cite{FordeJac09}.
\end{proof}
\subsection{The fat-tail case}\label{sec:FatTail}
If~$\vp$ is infinite and~$\m$ is finite, Proposition~\ref{prop:limit_of_cgf} 
states that the only choice for the rescaling factor is~$h(t) = t^{1/2}$, 
but the form of the limiting rescaled cumulant generating function 
does not yield any immediate asymptotic estimates for the probabilities.
In this case, we impose the following assumption on the moment generating function
of~$\Vv$ in the vicinity of the upper bound~$\m$ of its effective domain:
\begin{assumption}\label{assu:fat-tail-2}
There exists $(\gamma_0,\gamma_1,\gamma_2,\omega) \in \RR^*\times\RR\times\RR\times\NN_+^*$, 
such that the following asymptotics hold for the cgf of~$\Vv$ as~$u$ tends to~$\m$ from below:
\begin{equation}\label{eq:fat-tail assu2}
\log\M_\Vv(u) = 
\left\{
\begin{array}{ll}
\displaystyle \gamma_0\log(\m-u) + \gamma_1 + o(1), &\text{ for }\omega=1,\gamma_0<0,\\
\displaystyle \frac{\gamma_0}{(\m-u)^{\omega-1}} \left\{1+ \gamma_1(\m-u)\log(\m-u) + \gamma_2(\m-u) + o\left(\m-u\right)\right\}, & \text{ for }\omega\geq 2,\gamma_0>0,
\end{array}
\right.
\end{equation}
and
\begin{equation}\label{eq:fat-tail assu3}
\frac{\M'_\Vv(u)}{\M_\Vv(u)} =
\left\{
\begin{array}{ll}
\displaystyle \frac{|\gamma_0|}{\m-u}\left(1 + o(1)\right), &\text{ for }\omega =1,\gamma_0<0,\\
\displaystyle \frac{(\omega-1)\gamma_0 }{(\m-u)^\omega}
\left\{1+ \A(\m-u)\log(\m-u) + \B(\m-u) + o\left(\m-u\right)\right\}, &\text{ for }\omega\geq 2,\gamma_0>0,
\end{array}
\right.
\end{equation}
where~$\A :=\gamma_1(\omega-2)(\omega-1)^{-1}$ and
$\B :=\left[\gamma_2(\omega-2)-\gamma_1\right](\omega-1)^{-1}$.
\end{assumption}
\begin{remark}
Condition~\eqref{eq:fat-tail assu3} together with the expressions of~$\A$ and~$\B$ imply that 
the asymptotics of~$(\log(\M_\Vv))'$ can be derived by
differentiating~\eqref{eq:fat-tail assu2} term by term.
This is of course not always true;
however, Condition~\eqref{eq:fat-tail assu3} is rather mild, 
and we shall check it directly in several cases where~$\M_\Vv$ is known in closed form.
\end{remark}
\begin{example}\ 
\begin{itemize}
\item For the Exponential distribution with parameter~$\m$,
$(\gamma_0, \gamma_1, \omega) = (-1, \log \m, 1)$.
\item For the non-central~$\chi$-squared distribution as in Example~\ref{ex:NcCS},
$(\gamma_0,\gamma_1,\gamma_2,\omega) = \left(\frac{\lambda}{4}, -\frac{2q}{\lambda}, -2\left(1+\frac{q}{\lambda}\log 2\right), 2\right)$.
\end{itemize}
\end{example}
For $\m \in (0,\infty)$, introduce the function~$\Lambda^*:\RR\to\RR_+$ as
\begin{equation}\label{eq:fat-tail-call-asymp}
\Lambda^*(x) :=  \sqrt{2\m}|x|,
\end{equation}
as well as, for any $t>0$ the functions~$\Ee_t,\Cc_t:\RR^*\to\RR_+^*$ by
$\Ee_t(x) := \ind_{\{\omega=1\}} + \exp\left(\frac{c_1(x)}{t^{1/4}}\right)\ind_{\{\omega=2\}}$ and
\begin{equation}\label{eq:functionC}
\Cc_t(x) := 
\left\{
\begin{array}{ll}
\displaystyle
\exp\left(\frac{1}{2}(\rho\xi\m+1)x +\gamma_1 \right) \frac{|x|^{|\gamma_0|-1}}{\Gamma(|\gamma_0|)(2\m)^{1+|\gamma_0|/2}} t^{1-\frac{1}{2}|\gamma_0|}, &\text{ for }\omega=1,\\
\displaystyle
\exp\left(\frac{1}{2}(\rho\xi\m+1)x + \gamma_0\gamma_2 + \frac{\gamma_0}{4\m} \right)
\frac{1}{2\m\sqrt{2\pi}\zeta(x)}t^{\frac{7}{8} + \frac{1}{4}\gamma_0\gamma_1}, &\text{ for }\omega=2,
\end{array}
\right.
\end{equation}
where the functions~$c_1$ and~$\zeta$ are defined in 
Lemmas~\ref{lem:main-contribution}-\ref{lemma: cf-of-z} respectively.
Then the following behaviour, proved in Section~\ref{sec:proofOfFat-tailImpvolAsymps}, holds for European option prices:

\begin{theorem}\label{thm:fat-tailImpvolAsymps}
Under Assumption~\ref{assu:fat-tail-2}, 
European Call options with strike~$\E^x$ have the following expansion:
$$
\EE\left(\E^{X_t} - \E^x\right)^+
 =
\left(1-\E^x\right)^+ 
+ \exp\left(-\frac{\Lambda^*(x)}{\sqrt{t}}\right)\Ee_t(x)\Cc_t(x)\left(1+o(1)\right),
\quad\text{for any }x\ne 0,\text{ as }t\text{ tends to zero}.
$$
Moreover, the small-time implied volatility 
behaves as follows whenever~$x\neq 0$:
\begin{equation*}
\sigma_t^2(x) = \frac{|x|}{2\sqrt{2\m t}} + 
\left\{
\begin{array}{ll}
\displaystyle \mathrm{h}_1^{(1)}(x) + \mathrm{h}_2^{(1)}\log(t) + o(1), & \text{for }\omega=1,\\
\displaystyle \frac{c_1(x)}{4\m t^{1/4}} +\mathrm{h}_1^{(2)}(x) + \mathrm{h}_2^{(2)}\log(t) + o(1), & \text{for }\omega=2,
\end{array}
\right.
\end{equation*}
where 
\begin{align*}
\mathrm{h}_1^{(1)}(x) & := \frac{1}{4\m}\left\{\frac{\rho\xi\m }{2}x
+ \left(|\gamma_0|-\frac{1}{2}\right)\log|x| 
+ \gamma_1 + \log(4\sqrt{\pi}) 
- \left(\frac{|\gamma_0|}{2}+\frac{1}{4}\right)\log(2\m)-\log\Gamma(|\gamma_0|)\right\},\\
\mathrm{h}_1^{(2)}(x) & := 
 \frac{1}{4\m}\left\{\frac{\rho\xi\m}{2}x + \gamma_0\gamma_2 + \frac{9\gamma_0}{4\m} 
+\frac{5}{8}\log 2 -\frac{3}{8}\log \m + \frac{1}{4}\log \gamma_0 - \frac{1}{4}\log |x|\right\},\\
\mathrm{h}_2^{(1)} & := \frac{1}{8\m}\left(\frac{1}{2} - |\gamma_0|\right),
\qquad\qquad
\mathrm{h}_2^{(2)} := 
 \frac{1}{16\m}\left(\frac{1}{2} + \gamma_0\gamma_1\right),\quad\quad c_1(x) \text{ defined as in Lemma~\ref{lem:main-contribution}}.
\end{align*}
\end{theorem}
A particular example of a randomisation satisfying Assumption~\ref{assu:fat-tail-2} 
is the non-central Chi-squared distribution.
This case was the central focus of~\cite{jacquier13}, where the small-time behaviour 
of the forward smile in the Heston model was analysed.
As a sanity check, our theorem~\ref{thm:fat-tailImpvolAsymps} corresponds to~\cite[Theorem 4.1]{jacquier13}.

\begin{corollary}
Under Assumption~\ref{assu:fat-tail-2}, for $\omega\leq 2$,
$X\sim\LDP_0(\sqrt{t}, \Lambda^*)$.
\end{corollary}

\begin{remark}
Even though the leading order in the expansion is symmetric,
Theorem~\ref{thm:fat-tailImpvolAsymps} explains how 
the asymmetry in the volatility smile is generated.
In particular, the term~$\rho\xi x/8$ immediately
shows how the leverage effect can be produced with $\rho<0$.
\end{remark}

\subsection{Large-time asymptotics}\label{sec:LargeTime}
As observed in Figure~\ref{fig:appetiser},
the effect of initial randomness decays when the maturity becomes large,
so that the large-time behaviour of the randomised Heston model
should be similar to that of the standard Heston model,
which has been discussed in detail in~\cite{FordeJac11,FordeJacMij10,JKRM13}.
In the particular example of the forward Heston model--which coincides with randomising with a non-central $\chi$-squared distribution--such a large-time behaviour has been analysed in~\cite{JR15}.
Throughout this section we assume~$|\rho|<1$ and~$\kappa>\rho\xi$
(this condition usually holds on equity markets, 
where the instantaneous correlation~$\rho$ is negative--the so-called leverage effect), 
which guarantees the essential smoothness of the limiting cgf in a standard Heston as~$t$ tends to infinity, 
and define the function~$\Lf$ on~$\RR$ by:
\begin{equation}\label{eq:large-time-limiting-cgf}
\Lf(u)
:= 
\left\{
\begin{array}{ll}
\displaystyle\frac{\kappa\theta}{\xi^2}\left(\kappa-\rho\xi u - d(u)\right), &\text{ for }u\in[\usm, \usp],\\
+\infty, &\text{ for } u\in\RR\setminus[\usm, \usp],
\end{array}
\right.
\end{equation}
where~
$\usmp := \displaystyle\frac{1}{2\rrho^2\xi}\left(\xi-2\kappa\rho\pm\sqrt{(\xi-2\kappa\rho)^2 + 4\kappa^2\rrho^2}\right)$,  
and where the function~$d$ is given in~\eqref{eq:HestonMGF}. 
We further denote
$\Lf^*(x):=\sup_{u\in\RR}\left\{ux - \Lf(u)\right\}$, the convex conjugate of~$\Lf$. 
Forde and Jacquier~\cite[Theorem~2.1]{FordeJac11} proved that $\usm<0$ and $\usp>1$.
Consider now the following assumption:
\begin{assumption}\label{assu:Cond}
$\max\{\usm(\usm -1), \usp(\usp - 1)\} < \m\xi^2 \leq \infty$.
\end{assumption}

\begin{remark}
Assumption~\ref{assu:Cond} is a technical one, 
needed to ensure that the limiting cgf of the randomised model is essentially smooth.
Should it break down, a more refined analysis, similar to the one in~\cite{JR15} could be carried out 
to prove large deviations, but we leave it for future research.
\end{remark}
\begin{theorem}\label{thm:large-time-LDP}
Under Assumption~\ref{assu:Cond},
$(t^{-1}X_t)\sim\LDP_\infty(t^{-1}, \Lf^*)$ and 
\begin{equation*}
\lim_{t\uparrow \infty}\sigma_t^2(xt)
= 
\left\{
\begin{array}{ll}
2\left(2\Lf^*(x) - x + 2\sqrt{\Lf^*(x)(\Lf^*(x) -x)}\right), 
& \displaystyle \text{ for }x\in\left(-\frac{\theta}{2} , \frac{\ttheta}{2}\right),\\
2\left(2\Lf^*(x) - x - 2\sqrt{\Lf^*(x)(\Lf^*(x) -x)}\right), 
& \displaystyle  \text{ for }x\in\RR\setminus\left[-\frac{\theta}{2} , \frac{\ttheta}{2}\right],
\end{array}
\right.
\end{equation*}
where~$\ttheta := \frac{\kappa\theta}{\kappa-\rho\xi}>0$. 
If~$x\in\left\{-\frac{\theta}{2}, \frac{\ttheta}{2}\right\}$, then
~$\displaystyle\lim_{t\uparrow\infty}\sigma_t^2\left(\frac{\ttheta t}{2}\right) = \ttheta$, and
$\displaystyle\lim_{t\uparrow\infty}\sigma_t^2\left(-\frac{\theta t}{2}\right) =\theta.$
\end{theorem}
\begin{remark}\
\begin{itemize}
\item As proved in~\cite{FordeJac11}, the map 
$x\mapsto\Lf^*(x) - x$ is smooth, strictly convex, attains its minimum at the point~$\ttheta/2$, 
and~$\Lf^*(\ttheta/2)-\ttheta/2 = \Lf^*(\ttheta/2)'-1=0$. 
\item Theorem~\ref{thm:large-time-LDP} has the same form as~\cite[Corollary~2.4]{FordeJac11}, 
confirming the similar large-time behaviours of the classical and the randomised Heston models. 
\item Higher-order terms can be derived using 
the saddle point method described in detail in~\cite{FordeJacMij10}.
(see also~\cite[Proposition~2.12]{JR15}).
\end{itemize}
\end{remark}

Theorem~\ref{thm:large-time-LDP} provides the large-time behaviour of the implied volatility smile
with a time-dependent strike. 
For fixed strike, the initial randomisation has no effect,
and we recover the flattening effect of the smile:
\begin{corollary}[fixed strike]\label{coro: large-time-euro-call-put-2}
Under Assumption~\ref{assu:Cond}, 
$$
\lim_{t\uparrow \infty}\sigma_t^2(x)
 = 8\Lf^*(0)
  = \frac{4\kappa\theta}{\xi^2(1-\rho^2)}\left(-2\kappa + \rho\xi + \sqrt{\xi^2 + 4\kappa^2 - 4\kappa\rho\xi}\right),
\qquad\text{for all }x\in\RR.
$$
\end{corollary}

\subsection{At-the-money (ATM) case}\label{sec:ATM}
All our small-maturity results above hold in the out-of the money case $x\ne 0$.
As usual in the literature on implied volatility asymptotics, the at-the-money case 
exhibits a radically different behaviour, and a separate analysis is needed.
We first recall in Lemma~\ref{thm:cHestonATM} the at-the-money asymptotics in the classical Heston model~\cite{forde12}.
To differentiate between standard and randomised Heston models, 
denote by~$\sigma_t(x,v_0)$ as the implied volatility in the standard Heston model 
with fixed initial condition~$V_0 = v_0>0$.
\begin{lemma}~\cite[Corollary~4.4]{forde12}\label{thm:cHestonATM}
In the standard Heston model with~$V_0 = v_0>0$, 
assume that there exists~$\eps>0$ such that 
the map $(t,x)\mapsto\sigma_t^2(x,v_0)$ is of class~$\Cc^{1,1}([0,\eps)\times(-\eps,\eps))$, 
then 
$\sigma_t^2(0,v_0) = v_0 + a(v_0)t + o(t)$, where
$a(v_0) := -\frac{1}{12}\xi^2\left(1-\frac{1}{4}\rho^2\right) + \frac{1}{4}v_0\rho\xi + \frac{1}{2}\kappa(\theta-v_0)$.
\end{lemma}
\begin{theorem}\label{rHestonATM}
In a randomised Heston model, $\sigma_t(0) = \EE(\sqrt{\Vv}) + o(1)$ holds as time tends to zero.
\end{theorem}
\begin{proof}
Since $\m\in(0,\infty]$, then $\EE(\sqrt{\Vv})$ is finite.
Denote by $\Crm_{\BS}(t,x,\Sigma)$ the European Call option price in the Black-Scholes model
with maturity~$t$, strike~$\E^x$ and volatility~$\Sigma$, 
and by $\Crm_{\HH}(t,x,v)$ its price in the standard Heston model with~$V_0 = v$.
Using the tower property,
\begin{equation}\label{eq:ATMprice}
\EE\left(\E^{X_t} - 1\right)^+
= \EE\left(\EE\left(\E^{X_t} - 1\right)^+|\Vv\right) 
= \EE\left(\Crm_\HH(t,0,\Vv)\right),
\end{equation}
and Lemma~\ref{thm:cHestonATM} and~\cite[Corollary~4.5]{forde12} imply that
the equation $\Crm_\HH\left(t,0,\Vv\right) = \Crm_{\BS}\left(t,0,\sqrt{\Vv+a(\Vv)t}\right)\left(1+o(1)\right)$ holds~$\PP$-almost surely.
Also for any~$c\in\RR$,~\cite[Proposition~3.4]{forde12} implies that
$$
\Crm_{\BS}\left(t,0,\sqrt{\Sigma^2+ct}\right) 
= \frac{1}{\sqrt{2\pi}}\left(\Sigma t^{1/2} + \frac{12c/\Sigma-\Sigma^4}{24\Sigma}t^{3/2} + \Oo\left(t^{5/2}\right)\right).
$$
Plugging these equations back into~\eqref{eq:ATMprice}, 
and equating~\eqref{eq:ATMprice} with~$\Crm_{\BS}\left(t,0,\sigma_t(0)\right)$, 
the theorem follows from
$$
\sqrt{\frac{t}{2\pi}}\EE\left\{\left(\sqrt{\Vv} + \frac{12a(\Vv) - \Vv^{5/2}}{24\Vv}t + \Oo\left(t^2\right) \right)\left(1+o(1)\right)\right\}
= \sqrt{\frac{t}{2\pi}} \left(\sigma_t(0) - \frac{\sigma_t^3(0)}{24}t + \Oo\left(t^2\right)\right).
$$
\end{proof}
\begin{remark}\label{prop:ATMimproved}
If $\EE(\Vv^{-1/2})$ is finite then following a similar procedure
we obtain higher order terms of~$\sigma_t(0)$,
\begin{equation*}
\sigma_t(0) 
= \EE(\sqrt{\Vv}) + \left\{\cf_1\EE(\Vv^{-1/2}) + \cf_2\EE(\sqrt{\Vv}) + \cf_3\left(\EE(\sqrt{\Vv})^3 - \EE(\Vv^{3/2})\right)\right\}t + o(t),
\end{equation*}
where
$\cf_1 := \frac{1}{4}(\kappa\theta + \xi^2(\rho^2-4)/24)$,
$\cf_2 := \frac{1}{8}(\rho\xi-2\kappa)$, and
$\cf_3 := \frac{1}{24}$.
In the non-central chi-squared case we recover the result of~\cite[Theorem 4.4]{jacquier13}.
\end{remark}

\section{A dynamic pricing framework}\label{sec:dynamic}
The model proposed in this paper has so far only been studied in a static way,
namely from the inception time of the (European contract), with a view towards calibration of the implied volatility surface.
While providing a better fit to short-maturity options by steepening the skew,
it is not obvious, however, how to use the model dynamically;
in particular, it is unclear how to choose the random initial value of the volatility process 
during the life of the contract,
should one be wishing to sell or buy the option, or for hedging purposes.
Mathematically, assume that at time zero the trader chooses an initial randomisation~$\Vv$ 
(or classically a Dirac mass at some positive point), 
and suppose that, at some later time $\that>0$, she needs to reprice the option
(with remaining maturity~$\tau$).
How should she choose the new initial random variable~$\Vv_{\that}$?
Since the variance process has continuous paths, a suitable choice of~$\Vv_{\that}$,
consistent with the dynamics of the variance, is obviously~$V_{\that}$, 
the solution of the SDE~\eqref{eq:RandomHestonSDE},
after running it from time zero to time~$\that$.
With an initial guess~$\Vv$ at time zero, then, at time~$\that$, conditional on~$\Vv$,
$\Vv_{\that}$ is distributed as~$\beta_{\that}\chi^2(q,\lambda)$,
where $\beta_{\that}: = \xi^2(1-\E^{-\kappa \that})/(4\kappa)$, and
$\chi^2(q, \lambda)$ is a non-central Chi-squared distribution 
with $q := 4\kappa\theta / \xi^2$ degrees of freedom
and non-centrality parameter
$\lambda:=4\kappa \Vv/(\xi^2(\E^{\kappa \that}-1))$.
From the tower property,
the moment generating function of~$\Vv_{\that}$ then reads
\begin{equation}\label{eq:mgf_at_t}
\M_{\Vv_{\that}}(u) 
= \EE\left[\EE\left(\E^{uV_{\that}}|\Vv\right)\right]
= \left(1-2\beta_{\that} u \right)^{-q/2} \M_\Vv\left(\frac{\exp(-\kappa \that)u}{1-2\beta_{\that} u}\right),
\end{equation}
for all $u\in\Dd^H_{\that} = \{u\in\RR: \M_{\Vv_{\that}}(u)<\infty\}$. Setting $\n_t:=1/(2\beta_{\that})$, we have
\begin{equation*}
\Dd^H_{\that}
  = (-\infty,\n_{\that}) \bigcap\left\{u\in\RR: \frac{\exp(-\kappa \that)u}{1-2\beta_{\that} u} \in \Dd_{\Vv}\right\}
 =  \left\{
  \begin{array}{ll}
   (-\infty,\n_{\that}), & \text{ if }\m=\infty,\\
  \displaystyle  (-\infty,\n_{\that}) \bigcap\left(-\infty, \frac{\m}{\E^{-\kappa \that} + 2\beta_{\that} \m}\right)
 = (-\infty,\n_{\that}^*), & \text{ if }\m<\infty.
 \end{array}
 \right.
\end{equation*}
where $\n_{\that}^* :=  \frac{\m \n_{\that}}{\m + \n_{\that}\exp(-\kappa \that)}\label{eq:btStar}$.
We now discuss the impact of different choices of~$\Vv$ at time zero
on the distribution of~$\Vv_{\that}$
and on the implied variance~$\sigma^2_\tau(x, \that)$ at time~$\that$ (for a remaining maturity~$\tau$).
We keep here the terminology introduced in Section~\ref{sec:Asymptotics} regarding the tail behaviour of~$\Vv$.

Before diving into the detailed analysis, 
we argue that~$\Vv_{\that}$ chosen this way should only serve 
as a candidate for the initial distribution at time~$\that$
and in practice should be recalibrated according to updated (noisy) market observations at time~$\that$. 
Market noises explain how
the distribution of~$\Vv_{\that}$ can deviate from the ergodic distribution: 
the impact of the (instantaneous) noises
can change the shape and parameterisation of the randomisation.
We further comment that understanding the choice of~$\Vv_{\that}$ is also useful from a model risk point of view: 
at time zero, it is important to understand and simulate the behaviours 
of model parameters at a given future time.
We show in this section, in our setting, that~$\Vv_{\that}$ can in fact only be fat tailed,
and therefore, for consistency, one should probably start with~$\Vv$ in the class of fat-tail distributions.

\subsection{The bounded-support case}
In this case, $\Dd^H_{\that} = (-\infty, \n_{\that})$;
the proof of Proposition~\ref{prop:limit_of_cgf} showed that
$\lim\limits_{u\uparrow \infty}u^{-1}\log\M_{\Vv}(u) = \vp$.
Combining this with~\eqref{eq:mgf_at_t}, we obtain, as~$u$ tends to~$\n_{\that}$ from below, 
$$
\log\M_{\Vv_{\that}}(u) 
= -\frac{q}{2}\log(1-2\beta_{\that} u) + \frac{\E^{-\kappa \that}\vp u}{1-2\beta_{\that} u } \left(1+ o(1) \right)
= \frac{q}{2}\log\left(\frac{\n_{\that}}{\n_{\that} - u}\right) + \frac{\E^{-\kappa \that}\vp \n_{\that} u}{\n_{\that}-u} (1+o(1))
= \frac{\E^{-\kappa \that}\vp\n_{\that}^2}{\n_{\that} - u}\left(1+o(1)\right),
$$
so that, at leading order, $\Vv_{\that}$ behaves asymptotically as a fat-tail distribution
as in Assumption~\ref{assu:fat-tail-2} with $\omega=2$.
In the particular case of a uniform distribution on $[\vm, \vp] \subset [0,\infty)$,
as~$u$ tends to~$\n_{\that}$ from below, we obtain
\begin{align*}
\log\M_{\Vv_{\that}}(u) 
&= \frac{q}{2}\log\left(\frac{\n_{\that}}{\n_{\that}-u}\right)
 + \n_{\that} \vp\E^{-\kappa \that}\left(\frac{\n_{\that}}{\n_{\that}-u} - 1\right)
  + \log\left(\frac{1-\exp\{(\vm-\vp)\E^{-\kappa \that}\n_{\that} u/(\n_{\that} - u)\}}{\E^{-\kappa \that}\n_{\that} u (\vp-\vm)}(\n_{\that} - u)\right)\\
&= \frac{\E^{-\kappa \that}\n_{\that}^2 \vp}{\n_{\that}-u}
\left[1 + \frac{\E^{\kappa \that}(2-q)}{2\n_{\that}^2 \vp}(\n_{\that} - u)\log(\n_{\that} - u)
 + \left\{ \frac{\E^{\kappa \that}}{\n_{\that} \vp}\log\left(\frac{\E^{\kappa \that}\n_{\that}^{q/2-2}}{\vp-\vm}\right)-1\right\}
 \frac{\n_{\that} - u}{\n_{\that}}  + o(\n_{\that} - u)\right].
\end{align*}
Hence in a uniform randomisation environment, at future time $\that$, 
the shape of the distribution of~$\Vv_{\that}$ depends both on~$\Vv$ 
and on the parameters~$\kappa,\theta,\xi$ that control the dynamics of the variance process.
Moreover from Theorem~\ref{thm:fat-tailImpvolAsymps},
the implied variance at time $\that$, denoted by~$\sigma^2_\tau(x, \that)$, 
has an explosion rate of~$\sqrt{\tau}$:
$$
\sigma_\tau^2(x,\that) = \frac{|x|\tau^{-1/2}}{2\sqrt{2\n_{\that}}} 
+ \frac{\sqrt{\vp |x|}}{2\E^{\kappa \that/2}}(2\n_{\that}\tau)^{-1/4} + o\left(\tau^{-1/4}\right),
\quad\text{for all $x\ne 0$, as }\tau \text{ tends to zero}.
$$

\subsection{The thin-tail case (Assumption~\ref{Assu:V0})}
Here again, $\Dd^H_{\that} = (-\infty, \n_{\that})$ and applying Lemma~\ref{lemma: log-mgf-asymptotics}
with $\log f\sim -l_1 v^{l_2}$, we have
\begin{align}\label{eq:thin-tail-dynamic}
\log\M_{\Vv_{\that}}(u)
&= \frac{q}{2}\log\left(\frac{\n_{\that}}{\n_{\that}-u}\right) 
+ l_1(l_2-1)\left(\frac{1}{l_1l_2}\right)^{\frac{l_2}{l_2-1}}\left(\frac{\E^{-\kappa \that}\n_{\that}^2}{\n_{\that}-u}\right)^{\frac{l_2}{l_2-1}}(1+o(1))\nonumber\\
&= \frac{q}{2}\log\left(\frac{\n_{\that}}{\n_{\that}-u}\right) + \frac{2^{\ggp-1}\ccp}{\ggp}\left(\frac{\E^{-\kappa \that}\n_{\that}^2}{\n_{\that} - u}\right)^{\ggp}\left(1+o(1)\right),
\end{align}
as~$u$ tends to~$\n_{\that}$ from below,
so that a thin-tail initial randomisation generates a fat-tail distribution for~$\Vv_{\that}$ at time~$\that$.
In light of~\eqref{eq:thin-tail-dynamic}, Assumption~\ref{assu:fat-tail-2} does not hold,
and hence~$\Vv_{\that}$ is neither of Gamma or non-central Chi-squared type. 
A case-by-case analysis depending on the distribution of~$\Vv$ is therefore needed in order 
to make the $o(\cdot)$ term in~\eqref{eq:thin-tail-dynamic} more precise.
\begin{example}[Folded-Gaussian randomisation]
When $f(v) \equiv \cf\E^{-l_1v^2}$,
straightforward computations yield
\begin{align*}
\log\M_{\Vv_{\that}}(u) 
&= \frac{q}{2}\log\left(\frac{\n_{\that}}{\n_{\that}-u}\right) + \frac{1}{4l_1}\left(\frac{\E^{-\kappa \that}\n_{\that} u }{\n_{\that} - u}\right)^2
+\log\left(\cf\sqrt{\frac{\pi}{l_1}}\right) + \log\left(\frac{1}{2} + \Phi\left(\frac{\E^{-\kappa \that}\n_{\that} u }{\sqrt{2l_1}(\n_{\that} - u)}\right)\right)\\
&=: \frac{\cf_0}{(\n_{\that} - u)^2} + \frac{\cf_1}{\n_{\that} - u} + \cf_2 -\frac{q}{2}\log(\n_{\that}-u) + o(1).
\end{align*}
We can obtain the small-time asymptotic expansion of the option price
using an approach similar to the proof of Theorem~~\ref{thm:fat-tailImpvolAsymps}.
Specifically, only Lemma~\ref{lemma: cf-of-z}
needs to be adjusted, and the rescaling factor is now~$\vartheta(\tau) = \tau^{1/6}$;
the main contribution to the asymptotics of out-of-the-money option prices is still given 
in Lemma~\ref{lem:main-contribution}.
Translating this into the asymptotics of the implied variance, we obtain, for small~$\tau$,
$$
\sigma_\tau^2(x, \that)
 = \frac{|x|}{2\sqrt{2\n_{\that}\tau}} + \frac{2|x|^{2/3}}{3(4l_1)^{1/3}} \frac{\exp\left(-\frac{2\kappa \that}{3}\right)}{\tau^{1/3}}
  + o(\tau^{-1/3}).
$$
\end{example}
\subsection{The fat-tail case}
In this case, $\Dd_{\that}^{H} = (-\infty, \n_{\that}^*)$.
Here we only discuss two special cases for~$\Vv$: 
the Gamma distribution,
and the (scaled) non-central~$\chi$-squared distribution.
\begin{example}[Gamma randomisation]
If $\Vv\equalDistrib\Gamma(\alpha,\m)$, then from~\eqref{eq:mgf_at_t}, we have
\begin{align*}
\log\M_{\Vv_{\that}}(u) 
&=  \frac{q}{2}\log\left(\frac{\n_{\that}}{\n_{\that}-u}\right) - \alpha\log\left(1-\frac{\E^{-\kappa \that}u\n_{\that}}{\m(\n_{\that} - u)}\right)
= -\alpha\log\left(\frac{\m\n_{\that} - (\m+\exp(-\kappa \that)\n_{\that})u}{\m(\n_{\that}-u)}\right)\\
&= -\alpha\log(\n_{\that}^*-u) + \alpha\log\left(\frac{\m(\n_{\that}-u)}{\m + \exp(-\kappa \that)\n_{\that}}\right)
+ \frac{q}{2}\log\left(\frac{\n_{\that}}{\n_{\that}-u}\right).
\end{align*}
Consequently~$\Vv_{\that}$ is still a fat-tail distribution satisfying Assumption~\ref{assu:fat-tail-2}
with~$\omega=1$, $\gamma_0=-\alpha$, 
while the upper bound of the support of the mgf now depends on 
both the initial distribution~$\Vv$ and on the evolution of the process (through~$\n_{\that}^*$).
A direct application of Theorem~\ref{thm:fat-tailImpvolAsymps}
further suggests that, for small enough $\tau>0$,
$$
\sigma_\tau^2(x, \that) = \frac{|x|}{2\sqrt{2\n^*_{\that} \tau}} + \mathrm{h}_1^{(1)}(x)  + \mathrm{h}_2^{(1)}\log(\tau) + o(1),\qquad \text{with }\mathrm{h}_1^{(1)},\mathrm{h}_2^{(1)}\text{ given in Theorem~\ref{thm:fat-tailImpvolAsymps}}.
$$
\end{example}
\begin{example}[Non-central~$\chi^2$ randomisation]
If $\Vv \equalDistrib \alpha \chi^2(\A,\B)$
then~$\m = 1/(2\alpha)$, and
\begin{align*}
\log\M_{\Vv_{\that}}(u) 
&= \frac{q}{2}\log\left(\frac{\n_{\that}}{\n_{\that}-u}\right) 
+ \left.\left(\frac{\alpha\B z}{1-2\alpha z } - \frac{\A}{2}\log(1-2\alpha z)\right)\right|_{z = \exp(-\kappa \that)u/(1-2\beta_{\that} u)}\\
&= \frac{\alpha\E^{-\kappa \that}\B \n_{\that}^* u }{\n_{\that}^*-u} - \frac{\A}{2}\log\left(\frac{\n_{\that}(\n_{\that}^*-u)}{\n_{\that}^*(\n_{\that} - u)}\right) 
+ \frac{q}{2}\log\left(\frac{\n_{\that}}{\n_{\that} - u}\right)\\
&= \frac{\alpha\E^{-\kappa \that}\B \n_{\that}^{*2}}{\n_{\that}^*-u} - \frac{\A}{2}\log(\n_{\that}^*-u) 
+\frac{q-\A}{2}\log\left(\frac{\n_{\that}}{\n_{\that} - \n_{\that}^*}\right) + \frac{\A}{2}\log \n_{\that}^* - \alpha\E^{-\kappa \that}\B \n_{\that}^* + \Oo(\n_{\that}^* - u),
\end{align*}
which satisfies~\eqref{eq:fat-tail assu2} in Assumption~\ref{assu:fat-tail-2}
with~$\omega=2$ as~$u$ tends to~$\n_{\that}^*$, with~$\n_{\that}^*$ playing the role of the boundary~$\m$, 
and $\gamma_0 = \alpha\B \n_{\that}^{*2}\E^{-\kappa \that}$.
As a result, the implied variance~$\sigma_\tau^2(x, \that)$ has an explosion rate 
of~$\sqrt{\tau}$ as~$\tau$ tends to zero, and its full asymptotic expansion is provided in Theorem~\ref{thm:fat-tailImpvolAsymps}.
\end{example}

This analysis shows that a suitable choice for~$\Vv_{\that}$, consistent with the dynamics of the variance process,
can actually depend on the initial randomisation at time zero, as well as the evolution of the variance.
Even though all three types of initial randomisation imply a fat-tail initial distribution at future time,
the generated small remaining-maturity implied volatility smiles are very different.
The folded-Gaussian (thin-tail) generates a steeper smile 
compared to the bounded support case;
a fat-tail distribution for~$\Vv$ generate an even steeper volatility smile at~$\tau$,
since the coefficient of the leading order is~$\n_{\that}^*$, which is strictly less than~$\n_{\that}$.
\begin{remark}
All distributions discussed in Section~\ref{sec:Asymptotics} generate a fat-tail distribution for~$\Vv_{\that}$.
However, should the assumptions in Section~\ref{sec:Asymptotics} break down,
this may not be true any longer: 
Equation~\eqref{eq:mgf_at_t} suggests that 
the mgf of~$\Vv_{\that}$ can be ill-defined whenever that of~$\Vv$ does not exist,
in the case of a Cauchy distribution for example.
That said, the study of the effective domain below~\eqref{eq:mgf_at_t} indicates
that, in our setting, only fat-tail distributions for~$\Vv_{\that}$ are possible.
\end{remark}

\begin{remark}
As $\that$ tends to zero, $\beta_{\that}$ also converges to zero, $\B_{\that}$ diverges to infinity,
and $\B^*_{\that}$ tends to~$\m$ (defined on Page~\pageref{eq:btStar}). 
Plugging these into the asymptotic behaviour developed in Section~\ref{sec:dynamic},
we recover the moment generating functions from Section~\ref{sec:Asymptotics} as well as
the asymptotics of the implied variance.
\end{remark}


\section{Examples and Numerics}\label{sec:unbounded}
We now choose some common distributions supported on a subset of~$[0,\infty)$ 
for the initial randomisation to illustrate the results in Section~\ref{sec:Asymptotics}.
We first start with the bounded support case, and provide rigorous justifications to the 
statements in Section~\ref{sec:appetiser}.
In Section~\ref{ex:UniformDist}, we consider a uniformly distributed initial variance,
with~$\vp$ finite, and provide full asymptotics of European Call prices.
The remaining sections are devoted to the unbounded support case;
specifically, 
Sections~\ref{ex:NcCS}-\ref{ex:ErgodicDist} correspond to the fat-tail case,
so that Theorem~\ref{thm:fat-tailImpvolAsymps} can be applied.
The thin-tail environment is illustrated in Section~\ref{ex:FoldedGaussian}
where the initial distribution satisfies Assumption~\ref{Assu:V0} with $l_2=2$.

\subsection{Uniform randomisation}\label{ex:UniformDist}
Assume that~$\Vv$ is uniformly distributed on 
$[\vm,\vp]$ with $0\leq \vm<\vp<\infty$,
then Corollary~\ref{Coro:BddSupport} provides the leading term of short-time implied volatility.
However, as will be shown in Section~\ref{sec: numerics}, 
the true volatility smile for small~$t$ is much steeper compared with the leading term,
so that higher-order terms shall be considered.
For any~$x\neq 0$, denote by~$u_{\vp}^*(x)$ the unique solution in~$(u_-,u_+)$ 
to the equation $x = \Lambda'(u)\vp$,
with~$\Lambda$ described in~\eqref{eq:lambda_H}.
From~\cite[Remark 2.1]{forde12}, existence and uniqueness of such a solution are straightforward,
and ~$u_{\vp}^*(x) \neq 0$ holds for any non-zero~$x$.
Introduce the function 
$\Urm: \RR^*\to\RR_+^*$ by
\begin{equation}\label{eq:U}
\Urm_{\vp}(x) 
:=
\exp\left\{\Drm_0^0(u_{\vp}^*(x))\vp + \Crm_0(u_{\vp}^*(x)) + x\right\},
\end{equation}
where the functions~$\Drm_0^0$ and~$\Crm_0$ are provided in~\eqref{eq:D0}-\eqref{eq:C0}.
From~\cite[Remark 3.2]{forde12}, the function~$\Urm$ is well defined on $\RR^*$. 
The following theorem is the main result of this section, 
and provides a detailed asymptotic behaviour of Call option prices as the maturity becomes small:

\begin{theorem}\label{thm:uniform-higher-order}
Under uniform randomisation, as~$t$ decreases to zero,
European Call option prices behave as
$$
\EE\left( \E^{X_t} - \E^x\right)^+
= \left(1-\E^x\right)^+ 
+ \exp\left(-\frac{\Lambda_{\vp}^*(x)}{t}\right) 
\frac{\Urm_{\vp}(x)t^{5/2}\left(1 + o(1)\right)}{(\vp-\vm)\Lambda(u_{\vp}^*(x))u_{\vp}^*(x)^2\sqrt{2\pi\vp\Lambda''(u_{\vp}^*(x))}},
\quad\text{for any }x\ne 0,
$$
where the function~$\Lambda_{\vp}^*$ was introduced in Corollary~\ref{Coro:BddSupport}.
\end{theorem}
\begin{remark}\label{rmk:uniform-higher-order}\
\begin{itemize}
\item The remainder is of order~$t^{5/2}$, 
instead of~$t^{3/2}$ as in both standard Heston and Black-Scholes models~\cite{forde12}.
This can also be seen at the level of the (asymptotic behaviour of) corresponding densities,
as noted in Remark~\ref{rem:FunnyPowers} below.
\item The asymptotics holds locally for any fixed log-strike~$x\neq 0$. 
The numerics indicate that
for small $t>0$, as~$x$ tends to zero, 
the asymptotics of option prices and volatility smile explode to infinity.
This is in contrast with the standard Heston case~\cite[Section~5]{forde12}.
\item Since the function~$\Lambda$ is strictly positive and strictly convex on $(u_-,u_+)\setminus\{0\}$ and
$u_{\vp}(x) \in (u_-,u_+)\setminus\{0\}$ for any $x\ne 0$, 
the quotient on the right-hand side is well defined.
\item In a Black-Scholes model we have (see~\cite[Corollary~3.5]{forde12})
$$
\EE\left(\E^{X_t} - \E^x\right)^+ = (1-\E^x)^+ + \frac{1}{\sqrt{2\pi}x^2}\exp\left(-\frac{x^2}{2\Sigma^2 t} + \frac{x}{2}\right)(\Sigma^2t)^{3/2}(1+\Oo(t)).
$$
Compare it with Theorem~\ref{thm:uniform-higher-order} then
we obtain the higher-order term in the expansion of the implied variance, 
as $t$ tends to zero,
$$
\sigma_t(x)^2 = \frac{x^2}{2\Lambda^*_{\vp}(x)}  
+ \frac{x^2t}{2\Lambda^*_{\vp}(x)^2} \log\left(\frac{\Urm_{\vp}(x)\exp(-x/2)(2\Lambda_{\vp}^*(x))^{3/2}t}{(\vp-\vm)\Lambda(u_{\vp}^*(x))u_{\vp}^*(x)^2\sqrt{\vp\Lambda''(u_{\vp}^*(x))x^2}}\right) + o(t).
$$
\end{itemize}
\end{remark}
\begin{proof}
The procedure is essentially the same as the proof of Theorem~\ref{thm:fat-tailImpvolAsymps}.
Applying Lemmas~\ref{lemma:Dasymp} and~\ref{lemma:Casymp},
the rescaled cgf of~$X_t$ for each~$t$ is given by
(with the same notations as in~\eqref{eq:LambdaH})
\begin{align}\label{eq:cgf-uniform-asympt}
\Lambda_t(u)
&:=\Lambda_1\left(t,\frac{u}{t}\right) 
= t\Crm\left(t,\frac{u}{t}\right) + t\log\left(\M_\Vv\circ\Drm\left(t,\frac{u}{t}\right)\right)
= t\Crm\left(t,\frac{u}{t}\right) + t\log\left(\frac{\E^{\vp \Drm(t,u/t)} - \E^{\vm \Drm(t,u/t)}}{(\vp-\vm)\Drm(t,u/t)}\right)\nonumber\\
&= \vp\Lambda(u) + t\left(\Crm_0(u) + \vp\Drm_0^0(u) - \log\left((\vp-\vm)\Lambda(u)\right)\right) + t\log t + \Oo\left(t^2\right).
\end{align}
For fixed~$x>0$ and small enough~$t>0$,
introduce the time-dependent probability measure~$\QQ_t$ by
$$
\frac{\D\QQ_t}{\D\PP}:= \exp\left(\frac{u_{\vp}^*(x) X_t - \Lambda_t(u_{\vp}^*(x))}{t}\right).
$$
Changing the measure, plugging~\eqref{eq:cgf-uniform-asympt} and rearranging terms yield
the following expression for the Call option price with strike~$\E^x$:
$$
\EE\left(\E^{X_t} - \E^x\right)^+
=\exp\left(-\frac{\Lambda_{\vp}^*(x)}{t}\right)
\Urm_{\vp}(x)
\frac{t\left(1+\Oo(t)\right)}{(\vp - \vm)\Lambda(u_{\vp}^*(x))}
\EE^{\QQ_t}\left[\exp\left(\frac{-u_{\vp}^*(x)(X_t - x)}{t}\right)\left(\E^{X_t - x}-1\right)^+\right].
$$
It is easy to show that for fixed~$t>0$, 
under~$\QQ_t$ the random variable $\left(\frac{X_t - x}{\sqrt{t}}\right)$ converges weakly to a Gaussian distribution. 
The rest of the proof is similar to Section~\ref{sec:proofOfFat-tailImpvolAsymps} and we omit the details.
\end{proof}

We now explain the steepness of the volatility smile
in the uncorrelated case $\rho=0$.
Using the at-the-money curvature formula for the implied volatility (in uncorrelated stochastic volatility models)
proved by De Marco and Martini~\cite[Equation~(2.9)]{demarco10},
we can write, for any $t>0$, 
\begin{equation}\label{eq:curvature-atm}
\left.\partial_x^2\sigma(t,x)^2\right|_{x=0}
 = \frac{2}{t}\left\{\sigma(t,0)\sqrt{2\pi t}\exp\left(\frac{\sigma(t,0)^2 t}{8}\right)p_t(0)-1\right\},
\end{equation}
where~$p_t$ is the density of the log-price process at time~$t$.
In the standard Heston model with the initial condition $V_0 = v_0\in(\vm,\vp)$,
such that~$\EE(\sqrt{\Vv}) = \sqrt{v_0}$,
the small-time asymptotics of the density reads~\cite[Section~5.3]{FrizGerholdPinter}
$$
p_t(x) = \exp\left(-\frac{\Lambda_{v_0}^*(x)}{t}\right)\frac{\Urm_{v_0}(x)}{\sqrt{2\pi v_0 \Lambda''(x)}} t^{-1/2}\left(1+o(1)\right),\quad \text{for any }x\ne 0,
$$
with the function~$\Urm$ defined in~\eqref{eq:U}.
Applying the saddle point method similar to the proof of~\cite[Theorem~3.1]{forde12},
the small-time asymptotics of the density in a randomised setting, 
denoted as~$\widetilde{p}_t$, has the expression
$$
\widetilde{p}_t(x) = \exp\left(-\frac{\Lambda_{\vp}^*(x)}{t}\right)\frac{\Urm_{\vp}(x)}{\sqrt{2\pi\vp\Lambda''(x)}}\frac{t^{1/2}\left(1+o(1)\right)}{(\vp-\vm)\Lambda(u_{\vp}^*(x))},\quad x \ne 0.
$$
\begin{remark}\label{rem:FunnyPowers}
Note the difference between the powers $t^{1/2}$ and $t^{-1/2}$ in the expressions 
for~$p_t$ and~$\widetilde{p}_t$ above.
Even if, in the bounded support case, the leading-order term is not affected by the randomisation,
the latter does act at higher order.
We leave a precise study of this issue to further research.
\end{remark}
The ratio~$p_t(x)/\widetilde{p}_t(x)$ then reads
$$
\frac{p_t(x)}{\widetilde{p}_t(x)}
= \frac{1}{t}\exp\left(-\frac{\Lambda_{v_0}^*(x) - \Lambda_{\vp}^*(x)}{t}\right)
\left(\frac{\vp}{v_0}\right)^{\frac{1}{2}}
\frac{\Urm_{v_0}(x)}{\Urm_{\vp}(x)}(\vp-\vm)\Lambda(u_{\vp}^*(x))\left(1+o(1)\right),\quad x \ne 0.
$$
It is easy to verify that
$\lim\limits_{x\downarrow 0 }\Urm_{v_0}(x) = \lim\limits_{x\downarrow 0 }\Urm_{\vp}(x) = 1$ and
$\lim\limits_{x\downarrow 0 }\Lambda(u_{\vp}^*(x)) = 0$.
Moreover, for any fixed $x\ne 0$,
$$
\partial_v\Lambda_v^*(x) = \partial_v\left[u_v^*(x)x - v\Lambda(u_v^*(x))\right]
= \left[x-v\Lambda'(u_v^*(x))\right] \frac{\partial u_v^*(x)}{\partial v} - \Lambda(u_v^*(x))
= - \Lambda(u_v^*(x)) < 0.
$$
Combining these results,
assume that the density at zero can be approximated by~$p_t(x)$ for small enough~$x>0$, 
then there exists $t^*>0$ small enough such that $p_t(x) / \tilde{p}_t(x)<1$ for all $t \in (0,t^*)$.
Plugging it back to~\eqref{eq:curvature-atm},
and noticing that (Section~\ref{sec:ATM})
$\sigma(t,0)\sim\EE[\sqrt{\Vv}] = \sqrt{v_0} \sim \sigma_t(0,v_0)$ holds as~$t$ tends to zero,
then the small-time curvature in a uniformly randomised Heston is much larger compared with that of a standard Heston,
implying a much steeper smile around the at-the-money. 
Figure~\ref{graph:steep} provides a visual help.

Finally, we mention that
the tail behaviour of the implied volatility in a uniformly randomised Heston model
is similar to that of the standard Heston.
To see this, notice that the moment explosion property in the standard Heston setting
is described in~\cite[Proposition~3.1]{andersen05}.
Specifically,
the explosion of the mgf of~$X_t$ is equivalent to the explosion of the function~$\Drm$ provided in~\eqref{eq:HestonMGF}.
Moreover, Equation~\eqref{eq:TowerProperty} suggests that it is still the case in the uniform randomised setting,
since~$\m$ is infinity.
Then the similarity of the tail behaviours follows from~\cite{lee04}
(see also~\cite{benaim09, caravenna15}).

\subsection{Non-central $\chi$-squared distribution}\label{ex:NcCS}
Assume that~$\Vv$ is non-central~$\chi$-squared distributed with~$q>0$ degrees of freedom and non-centrality parameter~$\lambda > 0$, so that its moment generating function reads
$$
\M_{\Vv}(u) = \frac{1}{(1 - 2u)^{q/2}}\exp\left(\frac{\lambda u}{1-2u}\right),
\qquad\text{for all }u \in \Dd_\Vv = (-\infty, 1/2),
$$
then $\vp$ is infinite and $\m = 1/2$.
By Proposition~\ref{prop:limit_of_cgf}, the only suitable scale function is 
$h(t)\equiv \sqrt{t}$, which corresponds exactly to the forward-start Heston model,
the asymptotics of which have been studied thoroughly in~\cite{jacquier13}.
Applying Equation~\ref{eq: bddlimit} and L'H\^opital's rule with 
$\M_{\Vv}'(u) = \M_{\Vv}(u)\left(\lambda(1-2u)^{-2} + q(1-2u)^{-1}\right)$ 
imply that, at the right endpoint $u = \sqrt{2\m} = 1$, as~$t$ tends to zero,
the pointwise limit
$$
\lim_{t\downarrow 0}\Lambda_{1/2}\left(t,\frac{1}{\sqrt{t}}\right) 
= \frac{4}{2-\rho\xi}\lim_{s\downarrow 0}\frac{s^2\M_{\Vv}'(1/2-s)}{\M_{\Vv}(1/2-s)} 
= \frac{\lambda}{2-\rho\xi}
$$
can be either finite or infinite.
In particular, since $\lambda>0$, the pointwise limiting rescaled cumulant generating function is not
continuous at the right boundary of its effective domain.
The cgf of~$\Vv$ satisfies Assumption~\ref{assu:fat-tail-2} with~$\omega=2$, 
then Theorem~\ref{thm:fat-tailImpvolAsymps} implies that we can recover~\cite[Theorem 4.1]{jacquier13}:
$$
\sigma_t^2(x) 
= \frac{|x|}{2}t^{-1/2} + \frac{\sqrt{\lambda |x|}}{2}t^{-1/4} + o\left(t^{-1/4}\right) ,\quad\text{as }t\text{ tends to zero}.
$$

\subsection{Folded Gaussian distribution}\label{ex:FoldedGaussian}
Assume that $\Vv\equalDistrib |\Nn(0,1)|$, 
then the density of~$\Vv$ reads
$$
f(v) = \sqrt{\frac{2}{\pi}}\exp\left(-\frac{1}{2}v^2\right),
\qquad\text{for all }v \in \Dd_{\Vv}=\RR_+,
$$
which satisfies Assumption~\ref{Assu:V0}. 
Simple computations yield
$\M_{\Vv}\left(z\right) = 2\exp\left(z^2/2\right)\Phi(z)$,
for any~$z\in\RR$,
where~$\Phi$ denotes the Gaussian cumulative distribution function. 
Therefore, Lemma~\ref{lemma:Dasymp} implies that for ~$\gamma\in(0,1)$,
$$
t^\gamma\log \M_{\Vv}\left(\Drm\left(t,\frac{u}{t^\gamma}\right)\right)
= \frac{u^4}{8}t^{2-3\gamma} + \frac{\rho\xi u^5}{8}t^{3-4\gamma} - \frac{u^3}{4}t^{2-2\gamma} + \Oo\left(t^{4-5\gamma}\right)+ \Oo\left(t^\gamma\right).
$$
If~$\gamma=1$, then 
$\M_{\Vv}(x\Lambda(u)) = 2\exp\left(\frac{1}{2}\Lambda^2(u)x^2\right)\Phi(x\Lambda(u))$, 
and hence
$$
t\log \M_{\Vv}\left(\Drm\left(t,\frac{u}{t}\right)\right)
= \frac{\Lambda^2(u)}{2 t} + \Drm_0^{0}(u) + \Oo(t).
$$
The limit is therefore non-degenerate if and only if~$h(t) = t^{2/3}$, 
in which case~$\Lambda_{2/3}(u) = \frac{1}{8}u^4$, for all $u \in \RR$,
and Theorem~\ref{thm:ThinTail} implies 
$$
\lim_{t\downarrow 0}t^{1/3}\sigma_t^2(x) = \frac{(2x)^{2/3}}{3},
\quad\text{for all }x\ne 0.
$$

\subsection{Starting from the ergodic distribution}\label{ex:ErgodicDist}
Remark~\ref{rmk:props-of-v} shows that 
the stationary distribution of a randomised Heston model
has the density
$$
f_\infty(v) = \frac{a^b}{\Gamma(b)} v^{b-1}\E^{-av},\quad\text{for }v>0,
$$
where $a := \frac{2\kappa}{\xi^2}$ and $b := a\theta$. 
Assume now that~$f_\infty$ is the density of~$\Vv$, 
so that the cgf of~$\Vv$ is given by $\log\M_\Vv(u) = -b\log\left(a-u\right) + b\log a$, with~$u<a=\m$.
Then Assumption~\ref{assu:fat-tail-2} is satisfied with $\omega=1$.
A direct application of Theorem~\ref{thm:fat-tailImpvolAsymps} implies that 
$$
\sigma_t^2(x) = \frac{\xi |x|}{4\sqrt{\kappa t}} + o\left(t^{-1/2}\right),\quad\text{for any }x\neq 0.
$$

\subsection{Other distributions}
The following table (more refined version of the one in Section~\ref{sec:intro}) 
presents some common continuous distributions for the initial variance, 
together with the corresponding parameters~$\vp, \m, l_2$.
In each case, we indicate (up to a constant multiplier) the short-time behaviour of the smile. 
\begin{table}[ht]
\caption{Some continuous distributions with support in~$\RR_+$}
\centering
\begin{tabular}{c c c c c c}
\hline\hline
Name & $\vp$ & $\m$ & $l_2$ & Behaviour of~$\sigma^2_t(x)$ ($x\ne 0$) & Reference\\ [0.5ex] 
\hline
Beta & $1$ & $\infty$ & \  & $x^2 / \Lambda_{1}^*(x)$ & Equation~\ref{eq:IVStandard}\\
Exponential($\lambda$) & $\infty$ & $\lambda<\infty$ & $1$ & $|x|t^{-1/2}$ & Theorem~\ref{thm:fat-tailImpvolAsymps}\\
$\chi$-squared & $\infty$ & $1/2$ & $1$ & $|x|t^{-1/2}$ & Theorem~\ref{thm:fat-tailImpvolAsymps}\\
Rayleigh & $\infty$ & $\infty$ & $2$ & $ x^{2/3} t^{-1/3}$ & Theorem~\ref{thm:ThinTail}\\
Weibull ($k>1$) & $\infty$ & $\infty$ & $k$ & $ (x^2/t)^{1/(1+k)}$ & Theorem~\ref{thm:ThinTail}\\ [1ex]
\hline
\end{tabular}
\label{table:nonlin}
\end{table}

\subsection{Numerics}\label{sec: numerics}
We present numerical results for the implied volatility surface for three types of initial randomisation: 
uniform ($\vp<\infty$), exponential ($\m < \vp = \infty$) and folded-Gaussian ($\m = \vp =\infty$). 
To generate these surfaces, 
we apply Fast Fourier Transform (FFT) methods~\cite{carr99} to derive a matrix of option prices,
and then compute the corresponding implied volatilities using a root-finding algorithm.
The Heston parameters are given by~$(\kappa, \theta, \xi, V_0, \rho) = (2.1, 0.05, 0.1, 0.06, -0.6)$,
which corresponds to a realistic data set calibrated on the S\&P options data.
In view of Theorem~\ref{rHestonATM},
parameters of~$\Vv$ are chosen to satisfy $\EE(\sqrt{\Vv}) = \sqrt{V_0}$, 
so that results of standard and randomised Heston models can be compared.

The numerics show that
the randomised Heston model provides a much steeper short-time volatility smile compared with 
the standard Heston model,
but this difference tends to fade away as maturity increases.
In the uniform case, Figure~\ref{graph:uniform} and~\eqref{eq:IVStandard} may seem contradictory at first,
since the former indicates steepness and the latter excludes explosion. 
There is no issue here, and in fact suggests that even though there is no proper explosion,
it is still possible to generate steep short-time volatility smiles in a randomised setting. 
In Figure~\ref{graph:short-time-gamma} we test higher-order terms in a Gamma randomisation scheme 
while the Heston parameters remain unchanged.

\begin{figure}[h!]
\includegraphics[scale=0.4]{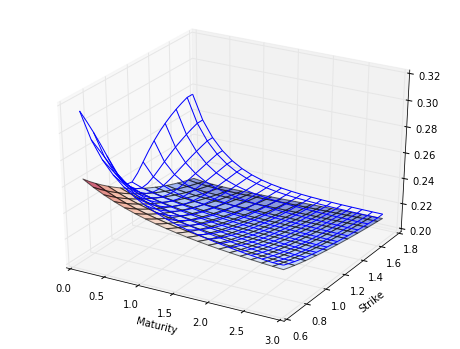}\quad
\includegraphics[scale=0.4]{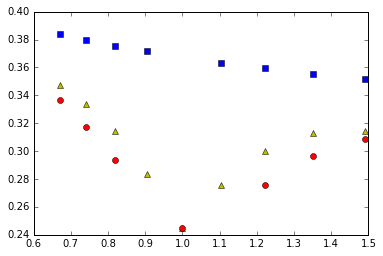}
\caption{Uniform randomisation with $(\vm, \vp) = (0,0.135)$.
Time-to-maturities are represented in years.
Left: volatility surfaces of randomised and standard Heston calculated with the FFT method. 
Right:
triangles, squares and circles represent implied volatility by FFT, leading-, and second-order asymptotics.
Time to maturity is $t = 1/24$. 
Higher-order terms are obtained by inverting the asymptotic formula in Theorem~\ref{thm:uniform-higher-order}
(see also Remark~\ref{rmk:uniform-higher-order}).}
\label{graph:uniform}
\end{figure}
\begin{figure}[h!]
\includegraphics[scale=0.4]{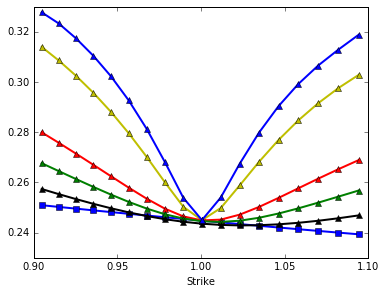}
\caption{Uniform randomisation with $(\vm, \vp) = (0,0.135)$.
Blue squares is the implied volatility obtained from the standard Heston model, maturity $t=0.005$.
Blue, yellow, red, green and black triangles represent the implied volatilities computed from the randomised Heston model by FFT, 
with the maturities equal to~$0.005$,~$0.01$,~$0.05$,~$0.1$, and~$0.2$.
The graph illustrates the increase of steepness in a randomised Heston setting as the maturity tends to zero.}
\label{graph:steep}
\end{figure}
\begin{figure}[h!]
\includegraphics[scale=0.385]{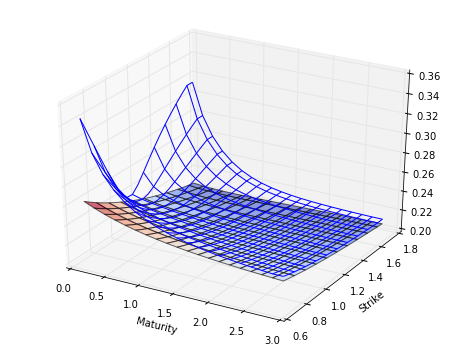}
\includegraphics[scale=0.385]{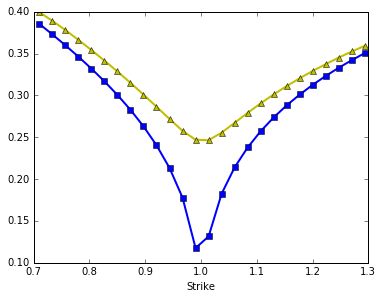}
\includegraphics[scale=0.385]{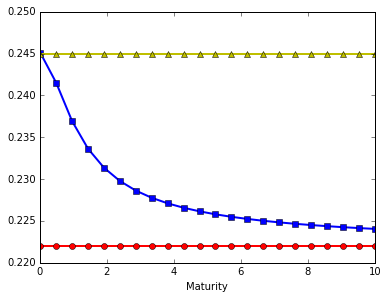}
\caption{$\Vv\equalDistrib$ Folded-Gaussian.
Left: implied volatility surfaces of folded-Gaussian randomisation and standard Heston,
calculated using FFT. 
Middle: implied volatility by FFT (triangles) and the leading order (squares) in Theorem~\ref{thm:ThinTail},~$t=1/24$.
Right: triangles, squares and circles represent
$\sqrt{V_0}$, ATM implied volatility~$\sigma_t(0)$ by FFT, and large-time limit.
The parameter~$l_1$ in Assumption~\ref{Assu:V0} is~$63.46$.}
\label{graph:short-time-fG}
\end{figure}
\begin{figure}[h!]
\includegraphics[scale=0.41]{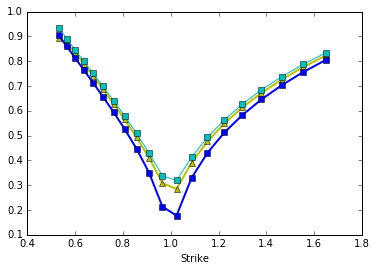}
\includegraphics[scale=0.41]{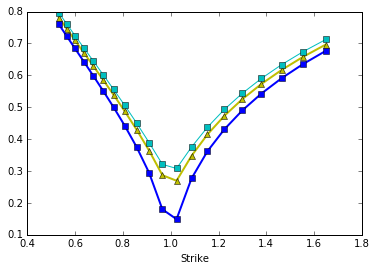}
\includegraphics[scale=0.41]{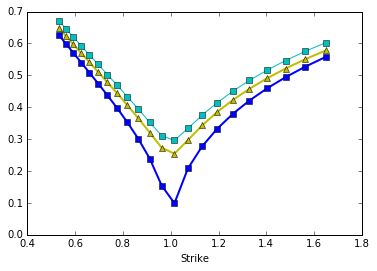}
\caption{$\Vv\equalDistrib\Gamma(\alpha,\beta)$ with $(\alpha,\beta) = (0.4, 3.868)$. 
Here we preset $\alpha$ and calculate $\beta$ using Theorem~\ref{rHestonATM}.
Blue and cyan squares are first- and second-order asymptotics,
yellow triangles are true smiles by FFT.
From left to right 
maturities are one week, two weeks and one month.}
\label{graph:short-time-gamma}
\end{figure}
In Figure~\ref{graph:dynamic-gamma} we illustrate the results in Section~\ref{sec:dynamic}.
We price the option 
in three different randomisation schemes after one month ($\that=1/12$) into the life of the contract. 
To compare different schemes, we again match the parameters of~$\Vv$ (at time zero) 
with different distributions
according to Theorem~\ref{rHestonATM}.
We see that the higher-order term in Theorem~\ref{thm:fat-tailImpvolAsymps}
is quite accurate even for relatively large time to maturity.
Not surprisingly (especially in the folded-Gaussian case)
the leading order is insufficient, and higher orders are needed for reliable approximations.
\begin{figure}[h!]
\includegraphics[scale=0.41]{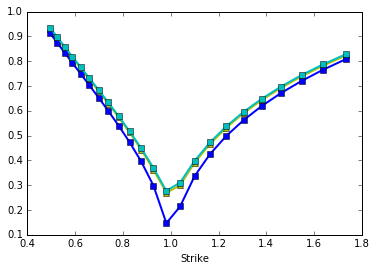}
\includegraphics[scale=0.41]{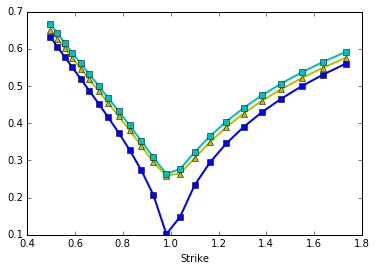}
\includegraphics[scale=0.41]{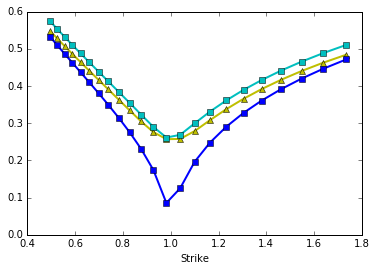}\\
\includegraphics[scale=0.41]{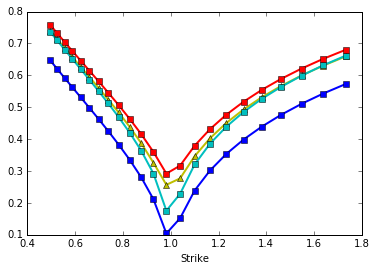}
\includegraphics[scale=0.41]{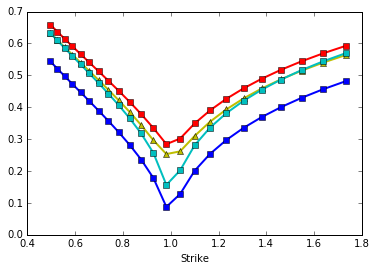}
\includegraphics[scale=0.41]{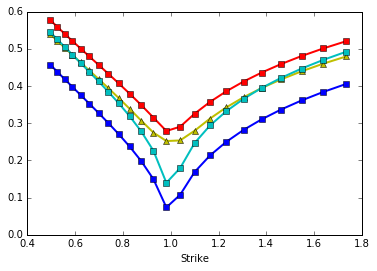}\\
\includegraphics[scale=0.41]{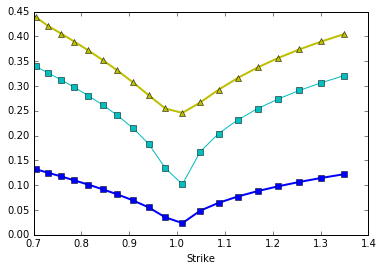}
\includegraphics[scale=0.41]{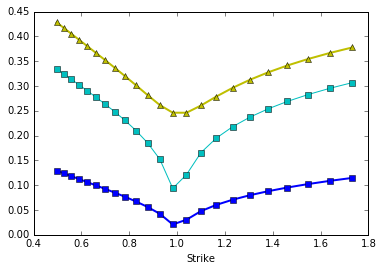}
\includegraphics[scale=0.41]{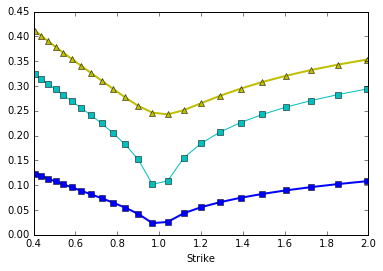}
\caption{
Numerical examples for the dynamic pricing framework.
We price the option at~$\that = 1/12$ in three different randomisation schemes:
Gamma ($\Gamma(0.4, 3.868)$, top), non-central Chi-squared ($0.07\chi^2(0.23, 1.25)$, middle),
folded-Gaussian ($l_2=2$, $l_1 = 63.46$, bottom).
Time to maturity $\tau$ is one week, one month, two months (left to right).
Blue and cyan squares are first- and second-order asymptotics,
red squares (in the second row) are third-order asymptotics;
yellow triangles are true implied volatilities computed by FFT.
}
\label{graph:dynamic-gamma}
\end{figure}
\clearpage
\subsection{USD/JPY FX options}
We test the calibration accuracy of the randomised Heston model 
using the USD/JPY FX market (ask) prices on January 20th, 2017.
In the FX market the implied volatility still has the small-time explosion feature:
Figure~\ref{fig:fx1} shows that the volatility smile generated by a standard Heston model
is too flat compared with the market data with small maturities.
This finding agrees with the existing literature.
For instance, in~\cite{janick10} the authors fixed~$\kappa$ and~$v_0$,
and calibrated the remaining~$3$ parameters
$(\theta,\xi,\rho)$ in a standard Heston environment to the EUR/USD market data.
They selected maturities ranging from one week to two years,
then calibrated the Heston model for each fixed maturity. 
Even with this `slice-by-slice' calibration procedure,
they observed poor fit of Heston to the market data for small maturities.
Unsurprisingly, they commented that time-dependent parameters, 
or `stochastic volatility plus jumps', as appeared in~\cite{bakshi97, Bates96}, 
are needed to improve the calibration accuracy.
We use the same initial guess for both the standard and the randomised Heston models, 
then calibrate the parameter sets using the market data.
The results are presented in Figure~\ref{fig:fx1} and Table~\ref{tab:CalibJPY}.
Both randomisation schemes have substantial improvement over the standard Heston model.
\begin{figure}[h!]
\includegraphics[scale=0.3]{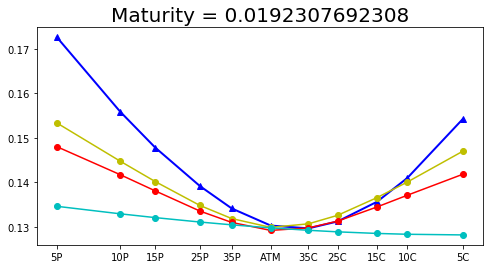}
\includegraphics[scale=0.3]{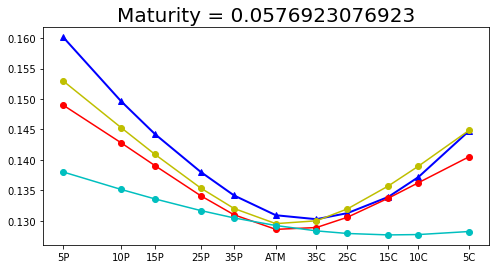}
\includegraphics[scale=0.3]{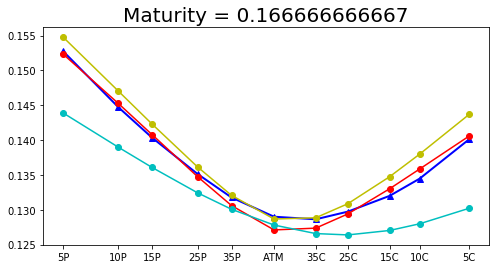}\\
\includegraphics[scale=0.3]{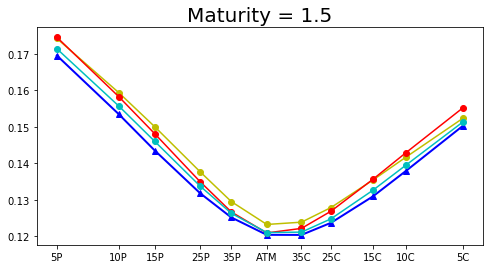}
\includegraphics[scale=0.3]{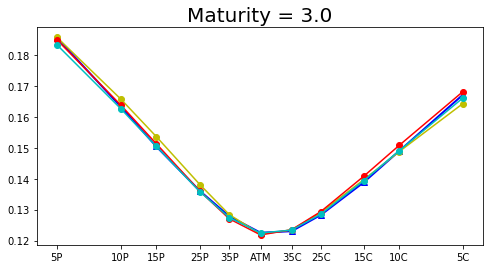}
\includegraphics[scale=0.3]{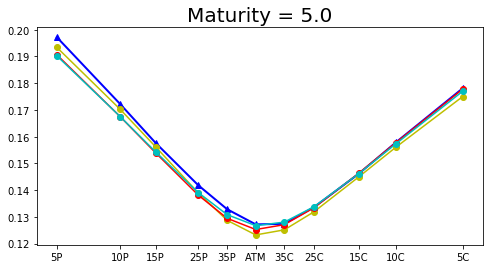}
\caption{
Calibration results of the randomised Heston model with Gamma (yellow) and uniform (red) randomisations, compared to the standard Heston model (cyan). 
}
\label{fig:fx1}
\end{figure}
\begin{table}[ht]
\centering
\begin{tabular}{c c c c}
\hline\hline
Model & Small maturities & Less than one year & Total\\ [0.5ex] 
\hline
Standard Heston & $11.91$ & $8.22$ & $7.34$\\
Gamma randomisation & $5.86$ & $5.02$ & $5.32$ \\
Uniform randomisation & $6.86$ & $5.13$ & $5.51$\\ 
\hline\\
\end{tabular}
\caption{The root-mean-square deviation (RMSD) of 
standard and randomised Heston models ($\times 10^{-3}$).
Small maturities are those less than one month;
the total RMSD is calculated over all maturities including extrapolations up to seven years.}
\label{tab:CalibJPY}
\end{table}

\appendix
\section{Notations from the Heston model}
In the Heston model, the log stock price satisfies the SDE~\eqref{eq:RandomHestonSDE},
where the initial distribution~$\Vv$ is a Dirac mass at some point~$v_0>0$.
As proved in~\cite{albrecher07}, 
the moment generating function~\eqref{eq:MGFX} admits the closed-form representation
$\M(t,u) = \exp\left(\Crm(t, u)+\Drm(t, u)v_0\right)$, for any $u\in\Dd_\M^t$,
where 
\begin{equation}\label{eq:HestonMGF}
\left\{
\begin{array}{rl}
\Crm(t,u) & := \displaystyle \frac{\kappa\theta}{\xi^{2}}
\left[(\kappa- \rho\xi u-d(u))t - 2\log\left(\frac{1-g(u)\E^{-d(u)t}}{1-g(u)}\right)\right],\\
\Drm(t,u) & := \displaystyle \frac{\kappa- \rho\xi u-d(u)}{\xi^{2}}\frac{1-\exp\left(-d(u)t\right)}{1-g(u)\exp\left(-d(u)t\right)},\\
d(u) & : = \displaystyle \left((\kappa - \rho\xi u)^{2}+\xi^{2} u (1-u)\right)^{1/2}
\qquad\text{and}\qquad
g(u) : = \frac{\kappa- \rho\xi u-d(u)}{\kappa-\rho\xi u +d(u)}.
\end{array}
\right.
\end{equation}
In the proof of~\cite[Lemma 6.1]{forde12}, the authors showed that the functions~$d$ and~$g$ have the following behaviour
as~$t$ tends to zero:
\begin{equation}\label{eq:heston_asymp}
d\left(\frac{u}{t}\right) = \I \frac{d_0 u}{t} + d_{1} + \Oo (t)
\qquad\text{and}\qquad
g\left(\frac{u}{t}\right) = g_{0} - \I \frac{g_1}{u}t + \Oo (t^2),
\end{equation}
with
$d_{0} := \xi\rrho\,\sgn(u)$, 
$d_{1} := \I \displaystyle \frac{2\kappa\rho-\xi}{2\rrho} \sgn(u)$, 
$g_{0} := \displaystyle \frac{\I\rho-\rrho\sgn(u)}{\I\rho+\rrho\,\sgn(u)}$ and 
$g_{1} := \displaystyle \frac{(2\kappa-\rho\xi)\sgn(u)}{\xi\rrho(\rrho+\I \rho\,\sgn(u))^{2}}$.
The pointwise limit of the (rescaled) cumulant generating function of~$X_t$ then reads
$$
\lim_{t\downarrow 0}t\log\M\left(t,\frac{u}{t}\right) = \Lambda(u)v_0,
\qquad\text{for any }u\in (u_-,u_+),
$$
where $u_-,u_+$ and $\Lambda$ are introduced in~\eqref{eq:lambda_H}.
From~\cite[Section 2]{forde12}, the function $\Lambda$ is well defined, 
smooth and strictly convex on $(u_-, u_+)$, and infinite elsewhere.

\section{Reminder on large deviations and Regular variations}\label{app:LDP}

\subsection{Large deviations and the G\"artner-Ellis theorem}\label{app:LDPLDP}
In this appendix, we briefly recall the main definitions and results from large deviations theory,
which we need in this paper. 
For full details, the interested reader is advised to look at the excellent monograph by Dembo and Zeitouni~\cite{dembo92}.
Let $(Y_n)_{n\geq 0}$ denote a sequence of real-valued random variables.
A map $I:\RR\to\RR_+$ is said to be a good rate function if it is lower semicontinuous and 
if the set $\{y:I(y)\leq \alpha\}$ is compact in~$\RR$ for each~$\alpha\geq 0$.
\begin{definition}
Let $h:\RR\to\RR_+$ be a continuous function that tends to zero at infinity.
The sequence~$(Y_n)_{n\geq 0}$ satisfies a large deviations principle as~$n$ tends to infinity with speed~$h(n)$ and rate function~$I$ 
(in our notations, $Y\sim\LDP_\infty(h(n), I)$)
if for each Borel measurable set~$\Ss\subset\RR$, the following inequalities hold:
$$
-\inf_{y\in\Ss^o}I(y)\leq\liminf_{n\uparrow\infty} h(n)\log\PP(Y_n\in\Ss)\leq\limsup_{n\uparrow\infty}h(n)\log\PP(Y_n\in\Ss)\leq -\inf_{y\in\overline{\Ss}}I(y).
$$
\end{definition}
Let now~$\Lambda_h$ be the pointwise limit of the rescaled cgf of~$Y$:
$\Lambda_h(u) := \lim_{n\uparrow\infty}h(n)\log\EE\left[\exp(uY_n/h(n))\right]$, whenever the limit exists,
and denote by~$\Dd_\Lambda:=\{u\in\RR: |\Lambda_h(u)|<\infty\}$ its effective domain.
Then~$\Lambda_h$ is said to be essentially smooth if
the interior~$\Dd_\Lambda^o$ is non-empty,~$\Lambda_h$ is differentiable on~$\Dd_\Lambda^o$, 
and $\displaystyle\lim_{u \to u_0} |\Lambda_h'(u) | = \infty,$ for any $u_0\in\partial\Dd_\Lambda$.
Finally, for any $y\in\RR$, define $\Lambda_h^*(y):=\sup_{u\in\Dd_\Lambda}\{uy - \Lambda_h(u)\}$, 
the convex conjugate of function~$\Lambda_h$.
\begin{theorem}[G\"artner-Ellis theorem, Theorem~2.3.6 in~\cite{dembo92}]\label{thm:Gartner}
If the function~$\Lambda_h$ is lower semicontinuous on~$\Dd_\Lambda$ and essentially smooth, and $0\in\Dd_\Lambda^o$,
then~$(Y_n)_{n\geq 0}\sim\LDP_\infty(h(n), \Lambda_h^*)$.
\end{theorem}

\subsection{Regular variations}\label{app:LDPRegVar}
We recall here some notions on regular variations, 
following the monograph~\cite{bingham89}.
\begin{definition}
Let $\A>0$.
A function~$f: (\A,\infty)\to\RR_+^*$  is said to be regularly varying with index~$l\in\RR$ (and we write $f\in\RV_l$)
if
$\lim\limits_{x\uparrow \infty}f(\lambda x)/f(x) = \lambda^l$, for any~$\lambda>0$.
When $l=0$, the function~$f$ is called slowly varying.
\end{definition}

\begin{lemma}[Bingham's Lemma, Theorem~4.12.10 in~\cite{bingham89}]\label{lem:Bingham}
Let~$f$ be a regularly varying function with index~$l>0$; 
then, as~$x$ tends to infinity, the asymptotic equivalence $-\log\int_x^\infty \E^{-f(y)}\D y \sim f(x)$ holds.
\end{lemma}
Let~$Y$ be a random variable supported on $[0,\infty)$ with a smooth density~$f$. 
The following lemma ensures that its moment generating function has unbounded support.
\begin{lemma}\label{lemma:suffi_condition}
If there exists~$l>1$ such that $|\log f| \in \RV_l$, then 
$\sup\{u\in\RR: \EE(\E^{uY})<\infty\} = +\infty$.
\end{lemma}
\begin{proof}
Karamata's Characterisation Theorem~\cite[Theorem~1.4.1]{bingham89} implies that 
 $|\log f(v)| = v^l g(v)$ for any~$v>0$, where the function~$g$ is slowly varying, and
Karamata's Representation Theorem~\cite[Theorem~1.3.1]{bingham89} provides the following expression:
$$
g(v) = \mathrm{c}(v)\exp\left(\int_a^v\varepsilon(y)\frac{\D y}{y}\right),
$$
where the functions~$\mathrm{c}$ and~$\varepsilon$ satisfy
$\lim_{v\uparrow \infty}\mathrm{c}(v) = c>0$,
$\lim_{v\uparrow \infty}\varepsilon(v) = 0$, and~$a$ is a fixed positive number. 
Then there exists~$v_1\geq a$ such that $\mathrm{c}(v)>c/2$ for all~$v\geq v_1$. 
Additionally, for any fixed~$\varepsilon_0$ small enough satisfying that $l>1+\varepsilon_0$, 
there exists $v_2\geq a$ such that
$\int_{v_2}^v\varepsilon(y)\D y/y > -\varepsilon_0\log (v/v_2)$, for any~$v\geq v_2$.
Denote $d:=\exp\left(\int_a^{v_2}\varepsilon(y)\frac{\D y}{y}\right)$, 
then for any $v>\max(v_1,v_2)$, and any $u>0$,
$$
u - \mathrm{c}(v)\exp\left(\int_a^v \varepsilon(y)\frac{\D y}{y}\right)v^{l-1}
<u - \frac{cd}{2}\exp\left(\int_{v_2}^{v}\varepsilon(y)\frac{\D y}{y}\right)v^{l-1}
<u - \frac{cd}{2}v_2^{\varepsilon_0} v^{l-1-\varepsilon_0}.
$$
Thus there exists~$v_3$ large enough so that
$u - \frac{1}{2}cdv_2^{\varepsilon_0} v^{l-1-\varepsilon_0}<-1$, for~$v\geq v_3$. 
With $v^*:=\max(v_1,v_2,v_3)$, 
$$
\EE\left(\E^{uY}\right)
=\int_0^{v^*}\E^{uv}f(v)\D v + \int_{v^*}^\infty \E^{v\left(u-v^{l-1} g(v)\right)}\D v
<\int_0^{v^*}\E^{uv}f(v)\D v + \int_{v^*}^\infty \E^{-v}\D v<\infty.
$$
\end{proof}

\section{Preliminary computations}
In view of~\eqref{eq:TowerProperty},
short-time asymptotic expansions of the functions~$\Crm$ and~$\Drm$ are necessary
in order to derive the pointwise limit of the rescaled cgf of~$(X_t)_{t\geq 0}$.
In this appendix we provide these expansions.
\subsection{Components of the mgf}\label{sec:components-of-mgf}
We start by investigating the short-time behaviour of the function~$\Drm(t,u/h(t))$.
For any $\beta \in \RR$, 
define the function~$\Drm_0^{\beta}:(u_-, u_+)\to\RR$ by
\begin{align}\label{eq:D0}
\Drm_0^{\beta}(u)
 & := 
\frac{1-\E^{-\I d_0 u}}{\xi^2 (1-g_0\E^{-\I d_0 u})}\left[(\rho\xi+\I d_0)\beta u+\kappa-d_1\right]
+ \frac{\I g_1(\rho\xi+\I d_0)}{\xi^2}
\frac{1-\E^{-\I d_0 u}}{(1-g_0\E^{-\I d_0 u})^2}\E^{-\I d_0 u}\\
&- \frac{(\rho\xi+\I d_0)u}{\xi^2}\frac{d_1-\I d_0 u\beta}{(1-g_0\E^{-\I d_0 u})^2}
\left(1-g_0\right)\E^{-\I d_0 u}\nonumber,
\end{align}
where the functions~$d_0, d_1, g_0, g_1$ are defined below~\eqref{eq:heston_asymp} above.

\begin{remark}\label{rem:RealPart}
The function~$\Drm_0^{\beta}$ is well defined:
to see this, we only need to check that the $\beta$ terms sum up to a real number, 
and the rest follows from~\cite[Remark~3.2]{forde12}. 
The first term in~\eqref{eq:D0} reads
$$
\frac{1-\E^{-\I d_0 u}}{\xi^2 (1-g_0\E^{-\I d_0 u})}(\rho\xi+\I d_0)\beta u
= -\beta\Lambda(u),
$$
which is a real number, and the sum of the remaining terms with~$\beta$ reads 
(taking out the prefactor $\I d_0 u \beta$)
$$
\frac{(\rho\xi + \I d_0)u\E^{-\I d_0 u}(1-g_0)}{\xi^2 (1-g_0 \exp(-\I d_0 u))^2}
= \frac{(g_0-1)\E^{-\I d_0 u}\Lambda(u)}{(1-g_0\E^{-\I d_0 u})(1-\E^{-\I d_0 u})}
= \frac{\I \rrho\,\sgn(u)\Lambda(u)}{\rho \cos(d_0 u)+\rrho\,\sgn(u) \sin(d_0 u)-\rho},
$$
which is purely imaginary, so that the whole term is a real number.
\end{remark}
The following lemma makes the effective domain of~$\Drm_0^{\beta}$ precise, 
and shows that it arises as the second order of the short-time expansion 
of a rescaled version of the function~$\Drm$ in~\eqref{eq:HestonMGF}.

\begin{lemma}\label{lemma:Dasymp}
Let~$\beta\in\RR$.
As~$t$ tends to zero, the map~$t\mapsto \Drm\left(t, u/h(t)\right)$ behaves as
\begin{equation*}
\Drm\left(t,\frac{u}{h(t)}\right) 
= 
\left\{
\begin{array}{lll}
0, & \text{if } u = 0, & \text{for any function }h,\\
\text{undefined}, &  u\neq 0, & \text{if } h(t)=o(t),\\
\displaystyle t^{-1}\Lambda(u) + \Drm_0^{\beta}(u) + o(1), & u\in(u_-, u_+), 
& \text{if }h(t) = t+\beta t^2+ o(t^2),\\
\displaystyle \frac{u^2 t}{2h^{2}(t)} \left[1 -\frac{h(t)}{u} + \frac{\rho\xi ut}{2h(t)}
 + \Oo\left(t + h^2(t) + \frac{t^2}{h^2(t)}\right) \right],  &u\in \RR, & \text{if } t=o(h(t)).
\end{array}
\right.
\end{equation*}
\end{lemma}
\begin{remark}\label{rmk: PowerScale}
\begin{enumerate}[(i)]
\item If~$h(t) = t + o(t)$ without further information on higher-order terms (third case in the lemma), 
then only the leading order is available:
$\Drm\left(t,u/h(t)\right) = t^{-1}\Lambda(u)(1+o(1))$.
\item As in Remark~\ref{rem:h}(ii), one can consider $h(t) = c t + \beta t^2 + o(t^2)$,
but by dilation, setting $c=1$ is inconsequent.
\item When $h(t) = t^{1/2}$,
$\Drm\left(t, \frac{u}{h(t)}\right) 
 = \frac{1}{2}u^2 + \frac{1}{4}\left(\rho\xi u^2 - 2\right)u t^{1/2} + \Oo(t)$,
 which is consistent with~\cite[Lemma 6.2]{jacquier13}.
\end{enumerate}
\end{remark}
The function~$\Crm_0:(u_-, u_+)\to\RR$ defined as
\begin{equation}\label{eq:C0}
\Crm_0(u)
:= -\frac{\kappa\theta}{\xi^{2}}\left[(\rho\xi+\I d_0)u+2\log\left(\frac{1-g_0\exp(-\I d_0 u)}{1-g_0}\right)\right]
\end{equation}
is clearly real valued~\cite[Remark 6.2]{forde12}, and determines the asymptotic behaviour of the function~$\Crm$
as follows.
\begin{lemma}\label{lemma:Casymp}
The map~$t\mapsto \Crm\left(t, u/h(t)\right)$ has the following asymptotic behaviour as~$t$ tends to zero:
\begin{equation*}
\Crm\left(t,\frac{u}{h(t)}\right) 
= 
\left\{
\begin{array}{lll}
\text{undefined}, & u\neq 0, & h(t)= o(t),\\
\displaystyle 
\Crm_0(u) + \Oo(t), & u\in(u_-, u_+),
& h(t) = t+\Oo(t^2),\\
\displaystyle \Oo\left(th(t) + h^3(t)\right)
 + \frac{\kappa\theta u^2}{4}\left(\frac{t}{h(t)}\right)^{2}
 \left[1 + \Oo\left(h(t) + \frac{t}{h(t)}\right)\right],  &  u\in \RR,
 & t=o(h(t)).
\end{array}
\right.
\end{equation*}
\end{lemma}

\begin{proof}[Proof of Lemma~\ref{lemma:Dasymp}]
Obviously $\Drm\left(t,0\right)\equiv 0$, so we assume from now on that $u \ne 0$.
From ~\eqref{eq:heston_asymp}, we have
\begin{equation}\label{eq:dgasymp}
d\left(\frac{u}{h(t)}\right) = \I \frac{d_0 u}{h(t)} + d_1 + \Oo\left(h(t)\right)
\qquad\qquad\text{and}\qquad\qquad
g\left(\frac{u}{h(t)}\right) = g_0 - \I \frac{g_1}{u}h(t) + \Oo\left(h^2(t)\right).
\end{equation}
Plugging these back into the expression of the function $\Drm$ in~\eqref{eq:HestonMGF}, we obtain
\begin{equation}\label{eq:Dasymp}
\Drm\left(t,\frac{u}{h(t)}\right)
 = \left[\frac{\kappa-d_1-\frac{\rho\xi u + \I d_0 u}{h(t)} + \Oo(h(t))}{\xi^2}\right]
\left[\frac{1-\exp\left\{-\frac{\I u d_0 t}{h(t)} - d_1 t+ \Oo(th(t))\right\}}
{1-\left[g_0 - \I \frac{g_1}{u}h(t) + \Oo(h^2(t))\right]
\E^{-\frac{\I u d_0 t}{h(t)} - d_1 t + \Oo(th(t))}}\right].
\end{equation}
If $h(t) = o(t)$, 
$d_0$ is a real number and $d_1$ is purely imaginary, then as $t/h(t)$ goes to infinity the term 
$\exp\left(-\I u d_0 t/h(t)-d_1 t\right)$ oscillates on the unit circle in the complex plane, 
thus no asymptotic can be derived.

Assume now that $h(t) = t + \beta t^2 + o(t^2)$. Then 
$th^{-1}(t) = 1-\beta t + o(t)$, and Equation~\eqref{eq:Dasymp} yields
\begin{align*}
\Drm&\left(t,\frac{u}{h(t)}\right) 
 = \frac{1}{\xi^2}\left[-\frac{(\rho\xi+\I d_0) u}{h(t)} + (\kappa-d_1) + \Oo(h(t))\right]
\left[\frac{1-\exp\left(-\I  d_0 ut/h(t)- d_1 t+ \Oo(th(t))\right)}{1-\left(g_0 -\I t g_1/u + \Oo(t^2)\right)
\exp\left(-\I  d_0 u t/h(t) - d_1 t + \Oo(th(t))\right)}\right]\\
&= \frac{1}{\xi^2}\left(-\frac{(\rho\xi+\I d_0) u}{t}\left(1-\beta t + o(t)\right) + \left(\kappa-d_1\right)+\Oo(t)\right)
\left(1-\E^{-\I d_{0}u}\left(1-d_1 t+\Oo(t^2)\right)\left(1+\I \beta d_0 u t + o(t)\right)\right)\nonumber\\
&\frac{1}{1-g_0\E^{-\I d_{0}u}}
\left(1+\frac{(-\I g_1/u + g_0(\I d_0 u \beta - d_1))\E^{-\I d_{0}u}}{1-g_0 \E^{-\I d_{0}u}}t
+ o(t)\right)\nonumber\\
&= \frac{\E^{-\I d_{0}u}-1}{\xi^2 t}\frac{(\rho\xi+\I d_0)u}{1-g_0\E^{-\I d_{0}u}}
+\frac{1-\E^{-\I d_{0}u}}{\xi^2 (1-g_0\E^{-\I d_{0}u})}\left((\rho\xi+\I d_0)\beta u+\kappa-d_1\right)
- \frac{(\rho\xi+\I d_0)u}{\xi^2}\frac{(d_1-\I\beta d_{0}u)\E^{-\I d_{0}u}}{1-g_0\E^{-\I d_{0}u}} \nonumber\\
&+ \frac{(\rho\xi+\I d_0)(\I g_1-g_0 u(\I\beta d_{0}u - d_1))(1-\E^{-\I d_{0}u})}
{\xi^2 (1-g_0\E^{-\I d_{0}u})^2}\E^{-\I d_{0}u} + o(1)\nonumber\\
&= \frac{\Lambda(u)}{t} + \Drm_0^{\beta}(u) + o(1).\nonumber
\end{align*}
The form of the effective domain is straightforward from these expressions.

If $h(t) = t + o(t)$ without further information on higher-order terms, 
then $t/h(t) = 1 + o(1)$. 
Following the same procedure as above then only the leading order can be derived, i.e.
$\Drm\left(t, u/h(t)\right) = t^{-1}\Lambda(u)[1+o(1)].$

Finally in the case $t = o(h(t))$,
\begin{align*}
 \left[1-\left(g_0 -\frac{\I g_1}{u}h(t) + \Oo(h^2(t))\right)
\E^{-\frac{\I d_0 ut}{h(t)}-d_1 t +\Oo(th(t))}\right]^{-1}
  & = \frac{1 - \frac{\I g_1 h(t)}{u(1-g_0)} - \frac{\I d_0g_0ut}{(1-g_0)h(t)}
 + \Oo\left(t + h^2(t) + \frac{t^2}{h^2(t)}\right)}{1-g_0},\\
 1-\exp\left(\frac{-\I d_0 ut}{h(t)}-d_1 t+\Oo(th(t))\right)
 &  = \frac{\I d_0 ut}{h(t)} + d_1 t + \frac{d_0^2u^2t^2}{2h^2(t)} + \Oo\left(\frac{t^2}{h(t)}\right) + \Oo\left(th(t)\right).
\end{align*}
Plugging these results into~\eqref{eq:Dasymp} yields
\begin{align*}
\Drm\left(t,\frac{u}{h(t)}\right)
& = \frac{1}{\xi^2(1-g_0)}
\left[-\frac{(\rho\xi+\I d_0) u}{h(t)} + (\kappa-d_1) + \Oo\left(h(t)\right)\right]
\left[\frac{\I d_0 ut}{h(t)} + d_1 t + \frac{d_0^2u^2t^2}{2h^2(t)}
 + \Oo\left(\frac{t^2}{h(t)} + th(t)\right)\right]\\
& \left[1 - \frac{\I g_1 h(t)}{u(1-g_0)} - \frac{\I d_0g_0ut}{(1-g_0)h(t)}
 + \Oo\left(t + h^2(t) + \frac{t^2}{h^2(t)}\right)\right]\\
& = \frac{1}{\xi^2(1-g_0)}\left[\frac{(d_0-\I \rho\xi)d_0u^2t}{h^2(t)}
 + \frac{\left[\I d_0 u(\kappa-d_1)-d_1u(\rho\xi+\I d_0)\right]t}{h(t)} -\frac{d_0^2u^3(\rho\xi+\I d_0)t^2}{2h^3(t)}
+\Oo\left(t + \frac{t^2}{h^2(t)}\right)\right]\\
& \left[1 - \frac{\I g_1 h(t)}{u(1-g_0)} - \frac{\I d_0g_0ut}{(1-g_0)h(t)}
 + \Oo\left(t + h^2(t) + \frac{t^2}{h^2(t)}\right)\right]\\
& = \frac{u^2t}{2h^2(t)}\left[1 -\frac{h(t)}{u} + \frac{\rho\xi ut}{2h(t)}
+\Oo\left(\frac{t^2}{h^2(t)} + h^2(t) + t\right)\right],
\end{align*}
where we used the identity
$$
\frac{(d_0-\I \rho\xi)d_0}{ (1-g_0)\xi^{2}} 
= \frac{(\xi\rrho \sgn(u)-\I \rho\xi)\xi\rrho \sgn(u)}{\xi^2\left(1-\frac{\I \rho-\rrho \sgn(u)}{\I \rho+\rrho \sgn(u)}\right)}
= \frac{1}{2}.
$$
\end{proof}

\begin{proof}[Proof of Lemma~\ref{lemma:Casymp}]
Assume that $u\ne 0$. 
Expand $d(u/h(t))$ and $g(u/h(t))$ to the third order,
\begin{equation}\label{eq:dg2order}
\begin{array}{rl}
\displaystyle d\left(\frac{u}{h(t)}\right) & = \displaystyle \I \frac{d_0 u}{h(t)}+ d_1 - \I d_2 h(t) + \Oo\left(h^2(t)\right),\\
\displaystyle g\left(\frac{u}{h(t)}\right) & = \displaystyle g_0 - \I \frac{g_1}{u}h(t) - \I \frac{g_2}{u^2}h^2(t) + \Oo\left(h^3(t)\right),
\end{array}
\end{equation}
where 
$d_2 := (\kappa^2 - d_1^2) / (2d_0 u)$ and 
$g_2 := [(\kappa^2-d_1^2)\rho\xi/d_0+(\kappa - d_1)(\rho\xi-\I d_0)g_1](\rho\xi-\I d_0)^{-2}$.
Combining these expansions with Equation~\eqref{eq:dgasymp} implies
\begin{align}\label{eq:Casymp1}
\Crm\left(t,\frac{u}{h(t)}\right) 
= &-\frac{2\kappa\theta}{\xi^2}
\log\left(\frac{1 - \left(g_0 - \I g_1h(t)/u-\I g_2 h^2(t)/u^2 + \Oo(h^3(t))\right)
\E^{-\I d_0 ut/h(t) - d_1 t + \I d_2 th(t)+ \Oo\left(th^2(t)\right)}}{1-g_0 + \I g_1h(t)/u + \I g_2 h^2(t)/u^2 +\Oo\left(h^3(t)\right)}\right)\nonumber\\
& + \frac{\kappa\theta}{\xi^2}\left[\left(\kappa-d_1\right)t - \frac{(\rho\xi+\I d_0)ut}{h(t)} + \Oo\left(th(t)\right)\right].
\end{align}

If $h(t) = o(t)$, no short-time asymptotics can be derived since $t/h(t)$ tends to infinity.
For the proof of the case where $h(t) = t+\Oo(t^2)$ we refer to~\cite[Lemma 6.1]{forde12}. 
Assume now that $t = o(h(t))$,
then the following asymptotic expansions hold:
\begin{align*}
\left(1-g_0+\frac{\I h(t)g_1}{u}+\frac{\I h^2(t)g_2}{u^2}+\Oo\left(h^3(t)\right)\right)^{-1}
&= \frac{1}{1-g_0}
\left(1-\frac{\I g_1 h(t)}{u(1-g_0)}-\frac{g_3 h^2(t)}{u^2(1-g_0)^2}+\Oo\left(h^3(t)\right)\right),\\
\exp\left(-\frac{\I d_0 u t}{h(t)}-d_1 t + \I d_2 th(t) + \Oo\left(th^2(t)\right)\right)
&= 1-\frac{\I d_0 ut}{h(t)}-\frac{1}{2}\left(\frac{d_0 ut}{h(t)}\right)^2 -d_1 t +\I d_2 th(t) 
+\Oo\left(th^2(t) + \frac{t^2}{h(t)}\right),
\end{align*}
where $g_3 := g_1^2+\I g_2(1-g_0)$.
Consequently,
\begin{align*}
&\frac{1-\left(g_0-\I g_1h(t)/u-\I g_2h^2(t)/u^2+\Oo(h^3(t))\right)
\E^{-\I d_0 ut/h(t)-d_1 t+\Oo\left(th(t)\right)}}{1-g_0 + \I g_1h(t)/u + \I g_2h^2(t)/u^2 + \Oo(h^3(t))}\\
&= \left\{1-\left[g_0-\frac{\I g_1}{u}h(t) - \frac{\I g_2}{u^2}h^2(t) + \Oo(h^3(t))\right]\right\}
\left[1-\frac{\I d_0 ut}{h(t)}-\frac{1}{2}\left(\frac{d_0 ut}{h(t)}\right)^2 -d_1 t + \I d_2 th(t) + \Oo\left(th^2(t)+\frac{t^2}{h(t)}\right)\right]\\
&\frac{1}{1-g_0}\left(1-\frac{\I g_1 h(t)}{u(1-g_0)}- \frac{g_3 h^2(t)}{u^2(1-g_0)^2} + \Oo(h^3(t))\right)\\
&= \left(1 + \frac{\I g_0 d_0 ut}{(1-g_0)h(t)} + \frac{\I g_1 h(t)}{u(1-g_0)} + \frac{d_1 g_0+d_0 g_1}{1-g_0}t + \frac{\I g_2 h^2(t)}{u^2(1-g_0)} + \frac{g_0 d_0^2 u^2 t^2}{2(1-g_0)h^2(t)}
+ \Oo\left(th(t)+h^3(t)+\frac{t^2}{h(t)}\right)\right)\\
&\left(1-\frac{\I g_1 h(t)}{u(1-g_0)}-\frac{g_3 h^2(t)}{u^2(1-g_0)^2}+\Oo\left(h^3(t)\right)\right)\\
&= 1+ \frac{\I g_0 d_0 u}{1-g_0}\frac{t}{h(t)}+\frac{u^2 d_0^2 g_0 t^2}{2(1-g_0)h^2(t)} + \left(\frac{d_1 g_0+d_0 g_1}{1-g_0}+\frac{d_0 g_0 g_1}{(1-g_0)^2}\right)t
+ \Oo \left(th(t)+h^3 (t)+\frac{t^2}{h(t)}\right),
\end{align*}
and therefore
\begin{align*}
&\log\left(\frac{1-\left[g_0-\I h(t)g_1/u-\I h^2(t)g_2/u^2+\Oo(h^3(t))\right]
\E^{-\I d_0 ut/h(t)-d_1 t+\Oo(th(t))}}{1-g_0 + \I h(t)g_1/u + \I h^2(t)g_2/u^2 +\Oo(h^3(t))}\right)\\
&= \frac{\I g_0 d_0 u}{1-g_0}\frac{t}{h(t)}
+ \frac{u^2 d_0^2 g_0}{2(1-g_0)^2}\frac{t^2}{h^2(t)}
+ \left(\frac{d_1 g_0+d_0 g_1}{1-g_0}+\frac{d_0 g_0 g_1}{(1-g_0)^2}\right)t 
+ \Oo\left(th(t) + h^3(t) + \frac{t^2}{h(t)} + \frac{t^3}{h^3(t)}\right)\nonumber.
\end{align*}
Plugging this into~\eqref{eq:Casymp1}, the result follows 
by noticing that the coefficients of~$\frac{t}{h(t)}$ and~$t$ are both zero.
\end{proof}



\section{Proofs of the main results}
\subsection{Proof of Proposition~\ref{prop:limit_of_cgf}}\label{subsec:proof_thm_limit_of_cgf}
In~\cite[Section~6]{jacquier13} the authors proved that $\Dd^* = \RR$ whenever~$\gamma<1$, 
and $\Dd^* = (u_-, u_+)$ if $\gamma=1$. 
Throughout the proof we keep the notation~$h$, emphasising that the statement still holds for function~$h$ with a general form, not only polynomials.

\textbf{Case $\gamma\in (0,1/2)$.}\label{sec:alphaLess12}
We need to analyse the behaviour of $\log \M(z)$ as~$z$ approaches zero. 
Since~$\m$ is strictly positive, by continuity of the mgf around the origin, 
$\M_{\Vv}\left(u^2 t(2h^2(t))^{-1}(1+\Oo(h(t)))\right)$ 
converges to $\M_{\Vv}(0) = 1$ as~$t$ tends to zero for any $u$ in $\RR$, which implies that $\Dd_{\Vv}^*= \RR$.
For small~$t$, a Taylor expansion indicates that
\begin{align*}
\log\M_{\Vv} \left(\Drm\left(t,\frac{u}{h(t)}\right)\right)
&= \log\EE\left(\exp\left\{\frac{u^2 t \Vv}{2 h^2(t)}
\left[1 -\frac{h(t)}{u}+\Oo\left(\frac{t}{h(t)}\right) + \Oo\left(h^2(t)\right)\right] \right\}\right)\\
&= \log\left\{1 + \frac{u^2\EE(\Vv) t}{2h^2(t)}
\left[1 -\frac{h(t)}{u} +\Oo\left(\frac{t}{h(t)}+h^2(t)\right)\right]
 + \frac{u^4 \EE(\Vv^2)t^2}{8h^4(t)} + \Oo\left(\frac{t^3}{h^6(t)}\right)\right\}\\
&= \frac{u^2 \EE(\Vv)t}{2h^2(t)}\left(1 +\Oo\left(h(t)+\frac{t}{h^2(t)}\right)\right).
\end{align*}
Since $h(t)\Crm(t,u/h(t))$ is of order $\Oo\left(t^2/h(t)+h^4(t)\right)$, then
\begin{equation}
\Lambda_\gamma\left(t,\frac{u}{h(t)}\right)
 = \frac{u^2 \EE(\Vv)t}{2h(t)}\left\{1 +\Oo\left(h(t)+\frac{t}{h^2(t)} + h^4(t)\right)\right\},
\end{equation}
and therefore
$\lim_{t\downarrow 0}\Lambda_\gamma(t, u/h(t)) = 0$, for all~$u\in\RR$.

\textbf{Case $\gamma \in (1/2,1]$.}\label{sec:alphaGreater12}
We need to evaluate~$\M_{\Vv}$ at infinity. 
If $\m$ is finite, for~$t$ sufficiently small, the term 
$\M_{\Vv}\left(\frac{1}{2}u^2 th^{-2}(t)\left(1+\Oo\left(t/h(t)\right)\right)\right)$ 
is infinite for any non-zero $u$,
hence $\Dd_{\Vv}^*=\{0\}$, and $\Lambda_\gamma(u)$ is null at $u=0$, 
and infinite elsewhere.	
If~$\m$ is infinite, then obviously~$\Dd_{\Vv}^* = \RR$. 
Assume first that~$\vp$ is finite; 
we claim that $\lim\limits_{u\uparrow \infty}(\vp u)^{-1}\log\M_{\Vv}(u) = 1$. 
In fact, let $F_{\Vv}$ be the cumulative distribution function of~$\Vv$, then
$$
\M_{\Vv}(u) = \EE\left(\E^{u\Vv}\right)\leq \exp(u\vp)\int_{[\vm, \vp]} F_{\Vv}(\D v) = \exp(u\vp).
$$
For any small $\eps>0$, fix $\delta \in(0,\eps\vp/2)$, so that
$$
\frac{\log\M_{\Vv}(u)}{u\vp}
\geq \frac{1}{u\vp}\log\left(\int_{\vp-\delta}^{\vp}\E^{uv}F_{\Vv}(\D v)\right)
\geq \frac{1}{u\vp}\log\left(\E^{u(\vp-\delta)}\PP\left(\Vv \geq \vp-\delta\right)\right)
= 1-\frac{\delta}{\vp} + \frac{\log\PP\left(\Vv \geq \vp-\delta\right)}{u\vp},
$$
since~$\vp$ is the upper bound of the support;
therefore $\PP(\Vv\geq \vp - \delta)$ is strictly positive, and the result follows.
If~$\gamma \in (1/2, 1)$, notice that 
$h(t)\Crm(t,u/h(t))$ is of order~$t^{2-\gamma}$ from Lemma~\ref{lemma:Casymp}, 
and hence
\begin{equation*}
\lim_{t\downarrow 0}\Lambda_\gamma\left(t, \frac{u}{h(t)}\right)
= \lim_{t\downarrow 0} \Oo\left(t^{2-\gamma}\right) 
+ \lim_{t\downarrow 0}
t^{\gamma}\log\M_{\Vv}\left(\frac{u^2 t^{1-2\gamma}}{2}\left(1+\Oo(t^{1-\gamma})\right)\right)
= \frac{u^2 \vp}{2} \lim_{t\downarrow 0} t^{1-\gamma}
= 0,\quad\text{for any }u\text{ in }\RR.
\end{equation*}
When $\gamma=1$, 
$\Lambda(u)$ is positive whenever $u\in(u_-, u_+)\setminus\{0\}$. 
Therefore,
\begin{equation*}
\lim_{t\downarrow 0}\Lambda_\gamma\left(t, \frac{u}{h(t)}\right) 
= \lim_{t\downarrow 0}\Oo(t) 
+ \lim_{t\downarrow 0} t\left(\frac{\vp\Lambda(u)}{t}(1+\Oo(t))\right)
= \Lambda(u)\vp,
\qquad\text{for any }u\in\Dd^* = (u_-, u_+).
\end{equation*}

\textbf{Case $\gamma=1/2$.}
If~$\vp$ is finite then the pointwise limit is null on the whole real line. 
Assume now that~$\vp$ is infinite and~$\m$ is finite. 
Following Remark~\ref{rmk: PowerScale}(iii),
$\frac{u^2}{2} + \left(\frac{\rho\xi u^3}{4}-\frac{u}{2}\right) t^{1/2}+\Oo(t) < \m$
implies
$\Dd_{\Vv}^{*o} = \left(-\sqrt{2\m}, \sqrt{2\m}\right)
\subseteq \Dd_{\Vv}^* \subseteq \limsup_{t\downarrow 0}\Dd_\Vv^t \subseteq\left[-\sqrt{2\m}, \sqrt{2\m}\right] = \overline{\Dd_{\Vv}^*}$.
For sufficiently small~$t$,
$$
\Lambda_{1/2}\left(t, \frac{u}{\sqrt{t}}\right) 
= \frac{\kappa\theta u^2}{4}t^{3/2} + \Oo\left(t^2\right) 
+ t^{1/2}\log\M_{\Vv}\left(\frac{u^2}{2}+\left(\frac{\rho\xi u^3}{4}-\frac{u}{2}\right) t^{1/2}+\Oo(t)\right).
$$
For any fixed~$u$ in $\Dd_{\Vv}^{*o}$, by definition there exists a positive $t_0$ such that~$u$ is in $\Dd_\Vv^t$ for all~$t$ less than $t_0$.
Then the mgf of~$\Vv$ is infinitely differentiable around the point $u^2/2$, and the~$n$-th order derivative at this point is 
$\M_{\Vv}^{(n)}\left(\frac{1}{2}u^2\right) = \EE\left[\Vv^n\exp\left(\frac{1}{2}u^2\Vv\right)\right]$.
Denote now 
$a_n (u) := \M_{\Vv}^{(n)}\left(\frac{1}{2}u^2\right)\M_{\Vv}^{-1}\left(\frac{1}{2}u^2\right)$, for $n\in \NN_+$, 
and 
$a_0(u):=\log\M_{\Vv}\left(\frac{1}{2}u^2\right)$.
A Taylor expansion of the function~$\M_{\Vv}$ around the point $\frac{1}{2}u^2$ yields
\begin{align}\label{eq: asymp-of-lamb-1/2}
\Lambda_{1/2}\left(t, \frac{u}{\sqrt{t}}\right)
& = \sqrt{t}\log\left\{\M_{\Vv}\left(\frac{u^2}{2}\right)\left[1+ a_1(u)
\left(\frac{\rho\xi u^2}{2}-1\right)\frac{u\sqrt{t}}{2} +\Oo(t)\right]\right\}
+ \frac{\kappa\theta u^2}{4}t^{3/2} + \Oo\left(t^2\right)\nonumber\\
&= a_0(u)\sqrt{t} + a_1(u)\left(\frac{\rho\xi u^2}{2}-1\right)\frac{ut}{2}
+\Oo(t^{3/2}).
\end{align}
Letting $t$ tend to zero, we finally obtain
\begin{equation*}
\Lambda_{1/2}(u)=
\left\{
\begin{array}{ll}
0, &\text{ when }u \in \Dd_{\Vv}^{*o},\\
\infty, & \text{ when }u\in \RR\setminus \overline{\Dd_{\Vv}^*}.
\end{array}
\right.
\end{equation*}
However, the limit of $\Lambda_{1/2}\left(t,\pm\sqrt{2\m/t}\right)$ depends on the explicit form of $\M_{\Vv}$. 
To see this, assume that $\rho\xi \m < 1$, 
which is guaranteed in particular when~$\rho\leq 0$, and compute the limit when $u = \sqrt{2\m}$. 
L'H\^opital's rule implies
\begin{equation}\label{eq: bddlimit}
\lim_{t\downarrow 0}t^{1/2}\log\M_{\Vv}\left(\m + \sqrt{\frac{\m}{2}}\left(\rho\xi \m-1\right)t^{1/2} +\Oo(t)\right) 
=\sqrt{\frac{2}{\m}}\frac{1}{1-\rho\xi \m}
\lim_{s\downarrow 0} \frac{s^2\M'_{\Vv}\left(\m - s\right)}{\M_{\Vv}\left(\m - s\right)}.
\end{equation}

\subsection{Proof of Theorem~\ref{thm:fat-tailImpvolAsymps}}\label{sec:proofOfFat-tailImpvolAsymps}
The systematic procedure is similar to the proof of~\cite[Theorem~3.1]{jacquier13}.
To simplify notations, write $\lambtilde_t(u) := \Lambda_{1/2}\left(t,u/\sqrt{t}\right)$,
$\ctilde_t(u):=\Crm(t,u/\sqrt{t})$ and
$\dtilde_t(u):=\Drm(t,u/\sqrt{t})$,
whenever these quantities are well defined.
We shall prove the theorem in several steps:
in Lemma~\ref{lemma:asymp-of-u*} we show that a saddlepoint analysis is feasible;
by taking the expectation under a new probability measure, the main contribution of the option price arises 
and its asymptotic expansion is provided in Lemma~\ref{lem:main-contribution};
in Lemma~\ref{lemma: cf-of-z} we prove the convergence (with rescaling) of the sequence~$(X_t-x)_{t\geq 0}$ under this new measure;
finally, the full asymptotics of the Call option price is obtained via inverse Fourier transform.

\begin{lemma}\label{lemma:asymp-of-u*}
Under Assumption~\ref{assu:fat-tail-2}, for any~$x\neq 0$,~$t>0$ small enough,
the equation~$\partial_u\lambtilde_t(u) = x$ admits a unique solution~$u_t^*(x)$
such that~$\dtilde_t(u_t^*(x))\in \Dd_\Vv^t$, 
and the following holds as $t$ tends to zero:
\begin{equation*}
u_t^*(x) = 
\left\{
\begin{array}{ll}
\displaystyle \sgn(x) \sqrt{2\m} + b_1(x)t^{\frac{1}{2\omega}} + o\left(t^{\frac{1}{2\omega}}\right), & \text{for }\omega= 1,\\
\displaystyle \sgn(x)\sqrt{2\m} + b_1(x)t^{\frac{1}{2\omega}} + b_2(x) t^{\frac{1}{\omega}}\log t + b_3(x)t^{\frac{1}{\omega}} + o\left(t^{\frac{1}{\omega}}\right), & \text{for }\omega\geq 2,
\end{array}
\right.
\end{equation*}
where
\begin{align*}
b_1(x) & := -\sgn(x) (2\m)^{(1-\omega)/(2\omega)}\left(\frac{\ind_{\{\omega=1\}}|\gamma_0| + \ind_{\{\omega\geq 2\}}(\omega-1)\gamma_0}{|x|}\right)^{1/\omega} 
+ \frac{1-\rho\xi\m}{2}\ind_{\{\omega=1\}},\\
b_2(x) & :=  -\sgn(x)\frac{\A\sqrt{2\m}}{2\omega^2}b_1^2(x),\\
b_3(x) & := \sgn(x)\left\{\frac{b_1^2(x)}{\sqrt{2\m}\omega}\left(1-\frac{\omega}{2} - 2\A\m\log\left(\sqrt{2\m}|b_1(x)|\right)-2\B\m\right)\right\}
+ \frac{1-\rho\xi\m}{2}\ind_{\{\omega=2\}}.
\end{align*}
If~$x=0$, then~$u_t^*(0)$ defined as the solution to~$\partial_u\lambtilde_t(u)=0$ 
satisfies $u_t^*(0) = \frac{1}{2}\sqrt{t} + o(\sqrt{t})$.
\end{lemma}

\begin{proof}[Proof of Lemma~\ref{lemma:asymp-of-u*}]
Assume that~$x>0$, the case when~$x<0$ being analogous.
Equation~\eqref{eq:TowerProperty} implies that for any~$u\in\RR$,
the equation~$\partial_u\lambtilde_t(u) = x$ reads
\begin{equation}\label{eq:fat-tail-sln}
x = \partial_u\lambtilde_t(u) = \sqrt{t}\left(\log\M\left(t,\frac{u}{\sqrt{t}}\right)\right)'
 = \sqrt{t}\ctilde_t'(u) 
+ \sqrt{t}\frac{\M'_\Vv\left(\dtilde_t(u)\right)}{\M_\Vv\left(\dtilde_t(u)\right)}\dtilde_t'(u).
\end{equation}
The existence and uniqueness of the solution to~\eqref{eq:fat-tail-sln} are guaranteed by the strict convexity of the rescaled cgf~$\lambtilde_t$ for each~$t$~\cite[Theorem~2.3]{jorgensen97} 
and~\eqref{eq:fat-tail assu3}, in which the denominator tends to zero as~$u$ tends to the boundary of~$\Dd_\Vv^t$.
Denote now the unique solution by~$u_t^*(x)$.
Applying Lemma~\ref{lemma:Dasymp} and Lemma~\ref{lemma:Casymp} with $h(t) \equiv t^{1/2}$, 
$$
\ctilde_t'(u) = \frac{u\kappa\theta}{2}t+\Oo\left(t^{\frac{3}{2}}\right),\quad\text{and}\quad
\dtilde_t'(u) = u + \left(\frac{3\rho\xi u^2}{4}-\frac{1}{2}\right)\sqrt{t} + \Oo(t).
$$
We first prove that~$\lim\limits_{t\downarrow 0} u^*_t(x) = \sqrt{2\m}$. 
If~$\lim_{t\downarrow 0}u_t^*(x)\neq \sqrt{2\m}$, 
there exists a sequence~$\{t_n\}_{n=1}^{\infty}$ and (small enough)~$\varepsilon_0>0$, satisfying~
$\lim_{n\uparrow \infty}t_n = 0$ and
$|u_{t_n}^*(x) - \sqrt{2\m}|\geq \varepsilon_0$ for any~$n\geq 1.$ 
In Section~\ref{subsec:proof_thm_limit_of_cgf} it is proved that
$\lim_{t\downarrow 0}\Dd_\Vv^t\subseteq \limsup_{t\downarrow 0}\Dd_\Vv^t\subseteq \overline{\Dd_{\Vv}^*} = [-\sqrt{2\m}, \sqrt{2\m}]$. 
Also notice that for any fixed~$t$ small enough the map~$\partial_u\lambtilde_t:\Dd_\Vv^t \to \RR$ is continuous and strictly increasing.
Hence for fixed positive~$\varepsilon_0$ there are at most finitely many~$t_i$ in the sequence such that~$u_{t_i}^*(x)\geq \sqrt{2\m}+ \varepsilon_0$.

Equation~\eqref{eq:fat-tail-sln} implies that for fixed~$x>0$ the limit of
$t^{-1/2}\partial_u\lambtilde_t(u_t^*(x))$
is infinity as~$t$ tends to zero.
Taking a subsequence of~$\{t_n\}_{n\geq 1}$ if necessary, 
assume now that $u_{t_n}^*(x)\leq \sqrt{2\m} - \varepsilon_0$, for any~$n\geq 1$.
Since~$\dtilde_{t}(\sqrt{2\m}-\eps) = \m-\sqrt{2\m}\eps + \eps^2/2 + \Oo(\sqrt{t})$,
then for any~$\eps > 0$ there exists~$N(\eps)\in\NN$ such that
$\left|\dtilde_{t_n}(\sqrt{2\m}-\eps) - \m + \sqrt{2\m}\eps - \eps^2/2\right|<\sqrt{2\m}\eps/2$
holds for any~$n\geq N(\eps)$.
Fix~$0<\eps_1<  \min(\eps_0,\sqrt{2\m})$ small enough so that 
$\m-3\sqrt{2\m}\eps_1/2 + \eps_1^2/2>\m-\delta_0$,
where~$\delta_0>0$ is chosen to satisfy for any~
$\m-\delta_0<u<\m$,
the higher-order term in~\eqref{eq:fat-tail assu3} is bounded above by one.
Then for such~$\eps_1$ and for any~$n\geq N(\eps_1)$ we have
$\m-\delta_0<\dtilde_{t_n}(\sqrt{2\m} - \eps_1)<\m-\sqrt{2\m}\eps_1/2+\eps_1^2/2<\m$.
The function~$\partial_u\lambtilde_t$ is strictly increasing, implying
$$
\lim_{n\uparrow \infty}\frac{\partial_u\lambtilde_{t_n}\left(u_{t_n}^*(x)\right)}{\sqrt{t_n}}
\leq \lim_{n\uparrow \infty}\frac{\partial_u\lambtilde_{t_n}\left(\sqrt{2\m} - \eps_1\right)}{\sqrt{t_n}}
\leq \frac{2^{\omega+1}\delta_1}{\eps_1^\omega(\sqrt{2\m}- \eps_1)^{\omega-1}}
<\infty,
$$
where
$\delta_1:= \ind_{\{\omega=1\}}|\gamma_0| + \ind_{\{\omega\geq 2\}}(\omega-1)\gamma_0$,
hence the contradiction. 
Therefore $\lim_{t\downarrow 0}u_t^*(x) = \sqrt{2\m}$. 
Analogously we can prove that $\lim_{t\downarrow 0}u_t^*(0) = 0$.

\textbf{Case~$\omega=1$.}
Assume that $u_t^*(x) = \sqrt{2\m} + h_x(t)$, where~$h_x(t) = o(1)$.
Equation~\eqref{eq:fat-tail-sln} implies that $h_x(t) = \Oo\left(\sqrt{t}\right)$,
hence all the terms of order~$\Oo(\sqrt{t})$ in the expansion of $\dtilde_t(u_t^*(x))$ should be included. 
More specifically,
$$
\dtilde_t\left(u_t^*(x)\right) = \m + \sqrt{2\m}h_x(t) + \frac{\sqrt{2\m}}{2}\left(\rho\xi\m-1\right)\sqrt{t} + o(\sqrt{t}).
$$ 
Plugging this back into~\eqref{eq:fat-tail-sln} 
and solving at the leading order yield the desired result.

\textbf{Case~$\omega\geq 2$.}
In this case $h_x(t) = \Oo(t^{1/(2\omega)})$.
Equation~\eqref{eq:fat-tail-sln} now reads
\begin{equation}\label{eq: solve-u*}
\frac{\sqrt{t}}{x}\left\{\frac{\kappa\theta\sqrt{2\m}}{2}t + \frac{\kappa\theta}{2}t h_x(t) + \Oo\left(t^{3/2}\right)
+ \left(\frac{\delta_0\left(1+o(1)\right)}{\left(-\sqrt{2\m}h_x(t) + \Oo\left(h_x^2+\sqrt{t}\right)\right)^\omega}\right)
\left(\sqrt{2\m} + h_x(t) + \Oo\left(\sqrt{t}\right)\right)\right\} \equiv 1.
\end{equation}
Denote by~$h_x^*$ as the leading order of the function~$h_x$.
Solving~\eqref{eq: solve-u*} at the leading order, then
$\delta_0\sqrt{2\m t} \equiv x\left(-\sqrt{2\m}h_x^*(t)\right)^\omega,$
from which
$h_x^*(t) = -\left(2\m\right)^{\frac{1-\omega}{2\omega}}\left(\delta_0/x\right)^{\frac{1}{\omega}} t^{\frac{1}{2\omega}}$.
Higher orders in the expansion of~$u_t^*(x)$ can be derived similarly, simply by
replacing the little-o term in~\eqref{eq: solve-u*} with precise higher-order terms provided in~\eqref{eq:fat-tail assu3}.
We omit the details.

Finally, when $x=0$,
write $u_t^*(0) = h(t)$ with $h(t) = o(1)$.
As~$t$ tends to zero,
$\M_\Vv(u_t^*(0))\sim 1$,~$\M_\Vv '(u_t^*(0)) \sim \EE(\Vv)$, and
$\dtilde_t'\left(u_t^*(0)\right) = h(t) - \frac{1}{2}\sqrt{t} + \Oo\left(t+h^2(t)\sqrt{t}\right)$.
Plugging these into~\eqref{eq:fat-tail-sln} with $x=0$ proves the lemma.
\end{proof}

\begin{lemma}\label{lem:main-contribution}\ 
\begin{enumerate}
\item When~$\gamma_0 > 0$ and~$\omega\geq 2$, as~$t$ tends to zero,
$$
\exp\left(\frac{-xu_t^*(x) + \lambtilde_t(u_t^*(x))}{\sqrt{t}}\right)
= \exp\left(-\frac{\Lambda^*(x)}{\sqrt{t}} + c_1(x) t^{\frac{1-\omega}{2\omega}} + o\left(t^{\frac{1-\omega}{2\omega}} \right)\right),
$$
for any~$x\neq 0$, where
$c_1(x) := \omega\gamma_0^{1/\omega}\left(\frac{|x|}{\sqrt{2\m}(\omega-1)}\right)^{1-1/\omega}$,
and the function~$\Lambda^*$ is defined in~\eqref{eq:fat-tail-call-asymp};
\item If~$\gamma_0<0$ and $\omega=1$, then for any $x\neq 0$, 
as~$t$ tends to zero,
$$
\exp\left(\frac{-xu_t^*(x) + \lambtilde_t(u_t^*(x))}{\sqrt{t}}\right)
= \exp\left(-\frac{\Lambda^*(x)}{\sqrt{t}} + c_2(x) + \gamma_1 \right)\left(\frac{|x|}{|\gamma_0|\sqrt{2\m t}}\right)^{|\gamma_0|}\left(1+o(1)\right),
$$
where~$c_2(x):= \frac{1}{2}\left(\rho\xi\m-1\right)x - \gamma_0 $.
\end{enumerate}
\end{lemma}
\begin{proof}[Proof of Lemma~\ref{lem:main-contribution}]
\textbf{Case~$\omega\geq 2$.}
Assumption~\ref{assu:fat-tail-2} and Lemma~\ref{lemma:asymp-of-u*} imply
\begin{equation*}
\begin{array}{lll}
\displaystyle \exp\left(-\frac{xu^*_t(x)}{\sqrt{t}}\right) 
 & \displaystyle = \exp\left\{-\frac{x}{\sqrt{t}}\left[\sqrt{2\m} + b_1(x) t^{\frac{1}{2\omega}} + o(t^{\frac{1}{2\omega}})\right]\right\}
 & \displaystyle = \exp\left\{-\frac{\Lambda^*(x)}{\sqrt{t}} - b_1(x) x t^{\frac{1-\omega}{2\omega}} + o\left(t^{\frac{1-\omega}{2\omega}} \right)\right\},\\
 \displaystyle\exp\left(\frac{\lambtilde_t(u_t^*(x))}{\sqrt{t}}\right)
 & \displaystyle = \exp\left(\ctilde_t(u^*_t) + \log\M_\Vv(\dtilde_t(u_t^*))\right)
 & \displaystyle = \exp\left\{\frac{\gamma_0}{\left(\sqrt{2\m}|b_1(x)|\right)^{\omega-1}}t^{\frac{1-\omega}{2\omega}} + o\left(t^{\frac{1-\omega}{2\omega}}\right)\right\}.
\end{array}
\end{equation*}
Using the expression of~$b_1(\cdot)$ provided in Lemma~\ref{lemma:asymp-of-u*}, 
the coefficient of the term of order~$t^{\frac{1-\omega}{2\omega}}$ is given by
$$
-b_1(x) x + \frac{\gamma_0}{\left(\sqrt{2\m}|b_1|\right)^{\omega-1}} 
= \left\{[(\omega-1)\gamma_0]^{\frac{1}{\omega}} + \gamma_0[(\omega-1)\gamma_0]^{\frac{1-\omega}{\omega}} \right\}\left(\sqrt{2\m}\right)^{\frac{1-\omega}{\omega}}|x|^{1-\frac{1}{\omega}}
= \omega \gamma_0^{\frac{1}{\omega}}\left(\frac{|x|}{\sqrt{2\m}(\omega-1)}\right)^{1-\frac{1}{\omega}}
= c_1(x) .
$$

\textbf{Case~$\omega=1$.}
This case follows straightforward computations after noticing that
$$
\exp\left\{\frac{\lambtilde_t(u_t^*(x))}{\sqrt{t}}\right\}
 = \exp\left\{\Oo(t) + \gamma_0\log \left(\m- \dtilde_t(u_t^*(x))\right) + \gamma_1 + o(1) \right\}
 =\E^{\gamma_1} \left(\frac{|\gamma_0|\sqrt{2\m t} }{|x|}\right)^{\gamma_0}\left(1+o(1)\right).
$$
\end{proof}

For each~$x\neq 0$ and~$t>0$ small enough, define the time-dependent measure~$\QQ_t$ by
$$
\frac{\D\QQ_t}{\D\PP} := \exp\left(\frac{u^*_t(x)X_t-\lambtilde_t(u_t^*(x))}{t^{1/2}}\right).
$$
Lemma~\ref{lemma:asymp-of-u*} implies that~$\lambtilde_t(u_t^*(x))$ is finite for small~$t$. 
Also by definition it is obvious that
$\EE[\D\QQ_t/\D \PP] = 1$, then~$\QQ_t$ is a well-defined probability measure for each~$t$.

\begin{lemma}\label{lemma: cf-of-z}
For any~$x\neq 0$, let~$Z_t:=(X_t-x)/\vartheta(t)$, where~
$\vartheta(t):=\ind_{\{\omega=1\}} + \ind_{\{\omega= 2\}}t^{1/8}$. 
Under Assumption~\ref{assu:fat-tail-2}, as~$t$ tends to zero, the characteristic function of~$Z_t$ under~$\QQ_t$ is
\begin{equation*}
\Psi_t(u) := \EE^{\QQ_t}\left(\E^{\I u Z_t}\right)
= 
\left\{
\begin{array}{ll}
\displaystyle\E^{-\I u x}\left(1 - \frac{\I u x}{|\gamma_0|}\right)^{\gamma_0}\left(1+o(1)\right), &\text{for }\omega =1,\\
\displaystyle \exp\left(\frac{-u^2\zeta^2(x)}{2}\right)\left(1+o(1) \right),&\text{for }\omega = 2,
\end{array}
\right.
\end{equation*}
where~$\zeta(x) := \displaystyle\sqrt{2}\left(\frac{2\m}{\gamma_0^2}\right)^{1/8}|x|^{3/4}$.
\end{lemma}
\begin{remark}
Lemma~\ref{lemma: cf-of-z} and L\'evy's Convergence Theorem~\cite[Theorem~18.1]{williams91} imply that 
under~$\QQ_t$ the process~$(Z_t)_{t\geq 0}$ converges weakly
 to a Gamma distribution (or a Gamma distribution mirrored to the negative real half line)
if $x>0$ (or $x<0$) minus the constant~$x$ when~$\omega=1$,
and to a Gaussian distribution when~$\omega=2$.
\end{remark}
\begin{remark}
Intuitively, the case $\omega\geq 3$ should be similar to the case $\omega=2$, 
so that a suitable candidate for the function~$\vartheta$ can be found.
However, in such scenario more information on the asymptotics of ~$\log\M_\Vv$ and its derivative 
are required in order to obtain the suitable (non-constant) characteristic function.
These extra assumptions turn out to be very restrictive and of little practical use, and are thus omitted.
\end{remark}

\begin{proof}[Proof of Lemma~\ref{lemma: cf-of-z}]
Assume that~$x>0$, with~$x<0$ being analogous. 
Function~$\log\Psi_t$ can be written as
\begin{align}\label{eq:cf-under-q}
\log\Psi_t(u) 
&= \log\EE\left[\exp\left(\frac{\I u (X_t-x)}{\vartheta(t)} + \frac{u_t^*(x)X_t - \lambtilde_t(u_t^*(x))}{\sqrt{t}}\right)\right]\\
&= -\frac{\I u x }{\vartheta(t)} + \log\EE\left[\exp\left(\frac{ \left(\I u\sqrt{t}/\vartheta(t) + u_t^*(x)\right)X_t}{\sqrt{t}}\right)\right] - \frac{\lambtilde_t(u_t^*(x))}{\sqrt{t}}\nonumber\\
&=-\frac{\I u x}{\vartheta(t)} + \frac{1}{\sqrt{t}}\left(\lambtilde_t\left(u_t^*(x) + \I u \sqrt{t}/\vartheta(t)\right) - \lambtilde_t\left(u_t^*(x)\right)\right).\nonumber
\end{align}

\textbf{Case~$\omega=1$.}
Lemma~\ref{lemma:asymp-of-u*} implies that
\begin{equation*}
\begin{array}{llrl}
& \displaystyle \dtilde_1(u)
:=\dtilde_t\left(u^*_t(x) + \frac{\I u\sqrt{t}}{\vartheta(t)}\right) 
= \m + \frac{\gamma_0\sqrt{2\m t}}{x} + \I u\sqrt{2\m t}
+ o(\sqrt{t}),
 &\dtilde_2
&\displaystyle :=\dtilde(u_t^*(x))
= \m + \frac{\gamma_0\sqrt{2\m t}}{x} 
+ o(\sqrt{t}),\\
& \displaystyle \ctilde_1(u)
:=\ctilde\left(u^*_t(x) + \frac{\I u\sqrt{t}}{\vartheta(t)}\right)
= \frac{\m\kappa\theta t}{2} + \Oo\left(t^{3/2}\right),
& \ctilde_2
& \displaystyle :=\ctilde(u^*_t(x)) = \frac{\m\kappa\theta t}{2}  + \Oo\left(t^{3/2}\right).
\end{array}
\end{equation*}
As a result, the lemma follows in this case from the following computations:
\begin{align*}
\log \Psi_t(u) 
&= -\I u x + \ctilde_1(u) - \ctilde_2 + \log\M_\Vv\left(\dtilde_1(u)\right) - \log\M_\Vv\left(\dtilde_2\right) 
= -\I u x + \gamma_0\log\left(\frac{\m-\dtilde_1(u)}{\m-\dtilde_2}\right) + o(1)\\
&= -\I u x + \gamma_0\log\left(1 - \frac{\I u x}{|\gamma_0|} + o(1)\right) + o(1).
\end{align*}

\textbf{Case $\omega = 2$.}
Denote~$\theta:=1/8$, then~$\frac{1}{2}-\theta>\frac{1}{4} = \frac{1}{2\omega}$.
Lemma~\ref{lemma:asymp-of-u*} implies
\begin{align*}
\dtilde_1(u)
&= \m +\sqrt{2\m}b_1 t^{1/4} + \I u\sqrt{2\m}t^{1/2-\theta} 
+ \left(\frac{b_1^2}{2} + \sqrt{2\m}b_3 + \sqrt{\frac{\m}{2}}(\m\rho\xi-1)\right)t^{1/2} 
+ o(\sqrt{t}),\\
\dtilde_2
&= \m +\sqrt{2\m}b_1 t^{1/4} + \left(\frac{b_1^2}{2} + \sqrt{2\m}b_3 + \sqrt{\frac{\m}{2}}(\m\rho\xi-1)\right)t^{1/2} 
+ o(\sqrt{t}),\\
\ctilde_1(u)
&= \frac{\m\kappa\theta t}{2} + \frac{\kappa\theta\sqrt{2\m}b_1}{2}t^{5/4} + \Oo\left(t^{11/8}\right),\quad\text{and}\quad
\ctilde_2  = \frac{\m\kappa\theta t}{2} + \frac{\kappa\theta\sqrt{2\m}b_1}{2}t^{5/4} + \Oo\left(t^{3/2}\right).
\end{align*}
Consequently,
\begin{align*}
&\frac{\lambtilde_t\left(u_t^*+\I ut^{1/2-\theta}\right) - \lambtilde_t\left(u_t^*\right)}{\sqrt{t}}
= \ctilde_1(u)-\ctilde_2+\log\M_\Vv(\dtilde_1(u))-\log\M_\Vv(\dtilde_2)\\
&= \frac{\gamma_0\left(\dtilde_1(u) - \dtilde_2\right)}{\left(\m-\dtilde_1(u)\right)\left(\m-\dtilde_2\right)}
+ \gamma_0\gamma_1 \left(\log\left(\m-\dtilde_1(u)\right) - \log\left(\m-\dtilde_2\right)\right) + o(1)\\
&= \frac{\gamma_0\left(\I u \sqrt{2\m} t^{1/2-\theta} + o(\sqrt{t})\right)}{2\m b_1^2t^{1/2}}
\left[1-\frac{\I u t^{1/4-\theta}}{b_1} + \Oo\left(t^{1/4}\right)\right]
+ \frac{\I u \gamma_0\gamma_1 t^{1/4-\theta}}{b_1} + o(1)
= \frac{\I \gamma_0 u}{\sqrt{2\m}b_1^2}t^{-\theta} + \frac{\gamma_0 u^2}{\sqrt{2\m}b_1^3}+o(1),
\end{align*}
and the proof follows by noticing that $b_1<0$ and $\gamma_0 = x \sqrt{2\m}b_1^2$, 
from Lemma~\ref{lemma:asymp-of-u*}.

\end{proof}

We finally prove the main theorem, when $x>0$. 
The price of a European Call option with strike~$\E^x$ is
\begin{align*}
\EE^\PP\left(\E^{X_t} - \E^x\right)^+
& = \EE^{\QQ_t}\left[\left(\E^{X_t} - \E^x\right)^+\frac{\D\PP}{\D \QQ_{t}}\right]
= \EE^{\QQ_t}\left[\exp\left(\frac{-u_t^*(x)X_t + \lambtilde_t(u_t^*(x))}{t^{1/2}}\right)\left(\E^{X_t} - \E^x\right)^+\right]\\
 & = \exp\left(\frac{-xu_t^*(x) + \lambtilde_t(u_t^*(x))}{\sqrt{t}}\right)\E^x
\EE^{\QQ_t}\left[\exp\left(\frac{-u_t^*(x)(X_t - x)}{\sqrt{t}}\right)\left(\E^{X_t-x} - 1\right)^+\right]\\
&= \exp\left(\frac{-xu_t^*(x) + \lambtilde_t(u_t^*(x))}{\sqrt{t}}\right)\E^x
\EE^{\QQ_t}
\left[\exp\left(\frac{-u_t^*(x)Z_t}{\sqrt{t}/\vartheta(t)}\right)\left(\E^{Z_t\vartheta(t)}-1\right)^+\right].
\end{align*}

\textbf{Case~$\omega=2$.}
The proof is identical to~\cite[Theorem~3.1]{jacquier13} and is therefore omitted.

\textbf{Case~$\omega=1$.} 
The Fourier transform of the modified payoff
$\exp\left(-\frac{u_t^*(x)Z_t}{\sqrt{t}}\right)\left(\E^{Z_t}-1\right)^+$ under~$\QQ_t$ is
\begin{align*}
\int_0^\infty \exp\left(-\frac{u_t^*(x)z}{\sqrt{t}}\right)\left(\E^z - 1\right)\E^{\I u z} \D z
&= \left[\frac{\E^{(1+\I u -u_t^*(x)t^{-1/2})z}}{1+\I u-u_t^*(x)t^{-1/2}}\right]_0^\infty 
- \left[\frac{\E^{(\I u -u_t^*(x)t^{-1/2})z}}{\I u-u_t^*(x)t^{-1/2}}\right]_0^\infty\\
&= \frac{1}{\I u-u_t^*(x)t^{-1/2}} - \frac{1}{1+\I u-u_t^*(x)t^{-1/2}} \\
&= \frac{t}{\left(u_t^*(x)-(1+\I u)\sqrt{t}\right)\left(u_t^*(x)-\I u\sqrt{t}\right)},
\end{align*}
where in the second line we use the fact that
$\lim_{t\downarrow 0} u_t^*(x)t^{-1/2} = +\infty$.
Recall that the Gamma distribution with shape~$|\gamma_0|$ and scale~$|\frac{x}{\gamma_0}|$ 
has density~$f_\Gamma \in L^2(\RR)$ given by
\begin{equation}\label{eq:density-f-gamma}
f_\Gamma(y) = \frac{y^{|\gamma_0|-1}}{\Gamma\left(|\gamma_0|\right)}
\exp\left(-\left|\frac{\gamma_0}{x}\right| y \right)\left(\left|\frac{\gamma_0}{x}\right|\right)^{|\gamma_0|},\quad\text{for }y>0.
\end{equation}
Applying~\cite[Theorem~13.E]{goldberg65} and Lemma~\ref{lemma: cf-of-z},
\begin{align}\label{inverse-to-density}
\EE^{\QQ_t}\left[\exp\left(\frac{-u_t^*(x)Z_t}{\sqrt{t}}\right)\left(\E^{Z_t}-1\right)^+\right]
&= \frac{t}{2\pi} \int_{-\infty}^\infty 
\frac{\Psi_t(u)\D u}{\left(u_t^*(x)-(1-\I u)\sqrt{t}\right)\left(u_t^*(x) + \I u\sqrt{t}\right)}\nonumber\\
&=\frac{t}{4\pi\m}\int_{-\infty}^\infty
\E^{-\I u x}\left(1-\frac{\I ux}{|\gamma_0|}\right)^{\gamma_0}\left(1+o(1)\right)\D u
=\frac{t f_\Gamma(x)}{2\m} \left(1+o(1)\right),
\end{align}
where the last line follows from Fourier inversion.
Combining Lemma~\ref{lem:main-contribution} and~\eqref{inverse-to-density},
the Call price reads
$$
\EE\left(\E^{X_t} - \E^x\right)^+
=
\exp\left(-\frac{\Lambda^*(x)}{\sqrt{t}}+ x + c_2(x) + \gamma_1\right)\left(\frac{x}{|\gamma_0|\sqrt{2\m}}\right)^{|\gamma_0|}
\frac{f_\Gamma(x)}{2\m} t^{1-\frac{|\gamma_0|}{2}}\left(1+o(1)\right),\quad\text{for } x>0.
$$

Assume now that~$x<0$, the price of a European Put option with strike~$\E^x$ is
$$
\EE\left(\E^x - \E^{X_t}\right)^+ = \exp\left(\frac{-xu_t^*(x) + \lambtilde_t(u_t^*(x))}{\sqrt{t}}\right)\E^x
\EE^{\QQ_t}
\left[\exp\left(\frac{-u_t^*(x)Z_t}{\sqrt{t}}\right)\left(1- \E^{Z_t}\right)^+\right],
$$
and the Fourier transform of the modified payoff function 
$\exp\left(\frac{-u_t^*(x)Z_t}{\sqrt{t}}\right)\left(1- \E^{Z_t}\right)^+$ is
$$
\int_{-\infty}^0 \exp\left(-\frac{u_t^*(x)z}{\sqrt{t}}\right)\left(1 - \E^z\right)\E^{\I u z} \D z
= \frac{t}{\left(u_t^*(x)-(1+\I u)\sqrt{t}\right)\left(u_t^*(x)-\I u\sqrt{t}\right)}.
$$
Following a similar procedure, and noticing that~$(\E^{X_t})_{t\geq 0}$ is a~$\PP$-martingale, the Put-Call parity implies
\begin{equation*}
\EE\left(\E^{X_t} - \E^x\right)^+
=
\left(1-\E^x\right)
+ \exp\left(-\frac{\Lambda^*(x)}{\sqrt{t}} + x + c_2(x) + \gamma_1\right)
\left(\frac{|x|}{|\gamma_0|\sqrt{2\m}}\right)^{|\gamma_0|}
\frac{f_\Gamma(|x|)}{2\m} t^{1-\frac{|\gamma_0|}{2}}\left(1+o(1)\right), \quad\text{for } x<0.
\end{equation*}
In the standard Black-Scholes model with volatility~$\Sigma>0$, the short-time asymptotics of the Call option price reads~\cite[Corollary~3.5]{forde12}
$\displaystyle\EE\left(\E^{X_t} - \E^x\right)^+ = (1-\E^x)^+ + \frac{1}{\sqrt{2\pi}x^2}\exp\left(-\frac{x^2}{2\Sigma^2 t} + \frac{x}{2}\right)
(\Sigma^2t)^{3/2}(1+\Oo(t))$.
Then the asymptotics of implied volatility can be derived following the systematic approach provided in~\cite{gao14}.

\subsection{Proof of Theorem~\ref{thm:large-time-LDP}}
We first prove the large deviations statement,
which we then translate into the large-maturity behaviour of the implied volatility.
Andersen and Piterbarg~\cite[Proposition~3.1]{andersen05} 
analysed moment explosions in the standard Heston model, and proved that for any $u>1$, 
the quantity $\EE(\E^{uX_t})$ always exists as long as 
\begin{equation}\label{assu:mgf-explosion}
\kappa>\rho\xi u\qquad\text{and}\qquad d(u)\geq 0.
\end{equation}
Moreover, the assumption $\kappa>\rho\xi$ implies (see~\cite{FordeJac11}) 
that~\eqref{assu:mgf-explosion} holds for any~$u\in[\usm, \usp]$, 
so that $\EE(\E^{uX_t})$ is well-defined for $u\in[\usm, \usp]$ 
and any (large)~$t$ in the standard Heston model.
The tower property then yields
$$
\M(t,u) = \EE\left[\EE(\E^{uX_t}|\Vv)\right] = \Crm(t,u)\left(\M_\Vv\circ\Drm(t,u)\right).
$$
Consequently, for any large~$t$, $\M(t,u)$ is well-defined for
$u\in \Ss:=[\usm, \usp]\cap\Ss_\Vv$, 
where the set~$\Ss_\Vv$ is defined by
$$
\Ss_\Vv:= \bigcup_{t>0}\bigcap_{s\geq t}\left\{u:\Drm(s,u)<\m \right\}.
$$
Using the expressions of functions~$\Crm$ and~$\Drm$ in~\eqref{eq:HestonMGF}, 
the rescaled cgf of the process~$(t^{-1}X_t)_{t\geq 0}$ reads
\begin{align}\label{eq:large-time-rescaled-cgf}
\Xi(t,u) 
&:= \frac{1}{t}\log\EE\left(\E^{uX_t}\right)
 = \frac{1}{t}\Crm(t,u) + \frac{1}{t}\log\M_\Vv\left(\Drm(t,u)\right)\nonumber\\
&= \frac{\kappa\theta(\kappa-\rho\xi u-d(u))}{\xi^2} 
- \frac{2\kappa\theta}{\xi^2 t}\log\left(\frac{1-g(u)\E^{-d(u)t}}{1-g(u)}\right)
+\frac{1}{t} \log\M_\Vv \left(\frac{\kappa- \rho\xi u - d(u)}{\xi^{2}}\frac{1-\E^{-d(u)t}}{1-g(u)\E^{-d(u)t}}\right).
\end{align}
For any~$u\in(\usm, \usp)$, since the quantity~$d(u)$ is strictly positive, then
\begin{equation}\label{eq:large-time-higher-limit}
\lim_{t\uparrow \infty}\frac{1}{t}\log\left(\frac{1-g(u)\E^{-d(u)t}}{1-g(u)}\right) = 0,\quad\text{and}\quad
\lim_{t\uparrow \infty}\frac{\kappa- \rho\xi u - d(u)}{\xi^{2}}\frac{1-\E^{-d(u)t}}{1-g(u)\E^{-d(u)t}} = \frac{\kappa- \rho\xi u - d(u)}{\xi^2}.
\end{equation}
Since~$u\mapsto\Xi(t,u)$ is continuous for each~$t>0$, L'H\^opital's rule implies that~
$\lim_{t\uparrow \infty}\Xi(t,\usmp) = \kappa\theta(\kappa - \rho\xi \usmp)/ \xi^2$.

\textbf{Case~$\m=\infty$.}
Obviously~$\Ss_\Vv = \RR$, implying that~$\Ss = [\usm, \usp]$.
Equation~\eqref{eq:large-time-rescaled-cgf} shows that
$$
\Xi(u):=\lim_{t\uparrow \infty}\Xi(t,u) 
\equiv \Lf(u),\quad\text{for any }u\in\RR,
$$
with~$\Lf$ provided in~\eqref{eq:large-time-limiting-cgf}.
In~\cite[Theorem~2.1]{FordeJac11}, it is proved that
the limiting function~$\Xi$ and its effective domain~$\Ss$ satisfy all the assumptions 
of the G\"artner-Ellis theorem (Theorem~\ref{thm:Gartner}), 
and hence the large deviations principle for the sequence~$(t^{-1}X_t)_{t\geq 0}$ follows.

\textbf{Case~$\m<\infty$.}
Equation~\eqref{eq:large-time-higher-limit} implies that
$$
\left\{ u: \frac{\kappa-\rho\xi u - d(u)}{\xi^2}<\m\right\}
\subset \Ss_\Vv
\subset \left\{ u: \frac{\kappa-\rho\xi u - d(u)}{\xi^2} \leq \m\right\}.
$$
As a result, the essential smoothness of function~$\Xi$ is guaranteed if
$$
[\usm, \usp]\subset \left\{ u: \frac{\kappa-\rho\xi u - d(u)}{\xi^2}<\m\right\}
= \left\{ u: \kappa-\rho\xi u <\xi^2\m + \sqrt{(\kappa-\rho\xi u)^2+\xi^2 u (1-u)}\right\}.
$$
Since~$\kappa-\rho\xi u>0$ holds for any~$u\in[\usm, \usp]$,
\begin{align}\label{condition-p}
\kappa-\rho\xi u <\xi^2\m + \sqrt{(\kappa-\rho\xi u)^2+\xi^2 u(1-u)}
 & \Longleftrightarrow
 0<u(1-u) + \xi^2 m^2 + 2\m \sqrt{(\kappa-\rho\xi u)^2+\xi^2 u(1-u)}\nonumber\\
 & \Longleftrightarrow
 \frac{u(u-1)}{\xi^2}<\m^2 + \frac{2\m}{\xi^2}\sqrt{(\kappa-\rho\xi u)^2+\xi^2 u(1-u)}.
\end{align}
Since $[0,1]\subset (\usm, \usp)$, 
Condition~\eqref{condition-p} holds for any~$u\in[0,1]$.
Whenever~$u>1$ or~$u<0$, as functions of~$u$,
the left-hand-side is strictly increasing while the right-hand-side is strictly decreasing.
Therefore,~\eqref{condition-p} holds for any~$u\in[\usm, \usp]$ if and only if
$\max\{\usm(\usm -1), \usp(\usp -1)\} < \m^2 \xi^2$.
Consequently, Assumption~\ref{assu:Cond} ensures that
$\Ss = [\usm, \usp]$, 
and the proof follows from the G\"artner-Ellis theorem (Theorem~\ref{thm:Gartner}).

We now prove the asymptotic behaviour for the implied volatility.
We claim that in a randomised Heston setting
the European option price has the following limiting behaviour:
\begin{align*}
-\lim_{t\uparrow \infty}\frac{1}{t}\log\left(\EE\left(\E^{X_t} - \E^{xt}\right)^+ \right) 
&= \Lf^*(x) - x, &\text{ for } x\geq\frac{\ttheta}{2},\\
-\lim_{t\uparrow \infty}\frac{1}{t}\log\left(1-\EE\left(\E^{X_t} - \E^{xt}\right)^+\right) 
&=\Lf^*(x) - x, &\text{ for } -\frac{\theta}{2}\leq x \leq \frac{\ttheta}{2},\\
-\lim_{t\uparrow \infty}\frac{1}{t}\log\left(\EE\left(\E^{xt} - \E^{X_t}\right)^+ \right) 
&=\Lf^*(x) - x, &\text{ for } x\leq-\frac{\theta}{2}.
\end{align*}
The proof is covered in detail in~\cite[Section~5.2.2]{JR15}, 
and we therefore highlight the main ideas for completeness.
From Theorem~\ref{thm:large-time-LDP}, define a time-dependent probability measure~$\QQ_t$:
$$
\frac{\D \QQ_t}{\D \PP} := \exp\left\{u^*(x)X_t - \Xi\left(t, u^*(x)\right)t\right\},
$$
where
$u^*(x)$ is the solution to the equation $x = \Xi'(u)$.
The option price is then expressed as the expectation under~$\QQ_t$ of a modified payoff, 
and can be computed by (inverse) Fourier transform
with the main contribution equal to~$\exp\left\{-\left(\Lf^*(x)-x\right)t\right\}$.
It is also known (see~\cite[Corollary~2.12]{FordeJac09} for instance) that
in the Black-Scholes model with volatility~$\Sigma$, 
the asymptotics of European option prices with strike~$\E^{xt}$ are given by
\begin{align*}
-\lim_{t\uparrow \infty}\frac{1}{t}\log\left(\EE\left(\E^{X_t} - \E^{xt}\right)^+ \right) 
&= \Lambda_{\BS}^*\left(x,\Sigma\right) - x, &\text{ for } x\geq\frac{\Sigma^2}{2},\\
-\lim_{t\uparrow \infty}\frac{1}{t}\log\left(1-\EE\left(\E^{X_t} - \E^{xt}\right)^+\right) 
&= \Lambda_{\BS}^*\left(x,\Sigma\right) - x, &\text{ for } -\frac{\Sigma^2}{2}\leq x \leq \frac{\Sigma^2}{2},\\
-\lim_{t\uparrow \infty}\frac{1}{t}\log\left(\EE\left(\E^{xt} - \E^{X_t}\right)^+ \right) 
&= \Lambda_{\BS}^*\left(x,\Sigma\right)- x, &\text{ for }x\leq-\frac{\Sigma^2}{2},
\end{align*}
where~$\Lambda_{\BS}^*\left(x,\Sigma\right):= \frac{\left(x+\Sigma^2/2\right)^2}{2\Sigma^2}$.
Then the leading order of the large-time implied variance is obtained by solving
$$
\Lf^*(x) - x= \Lambda_{\BS}^*(x,\Sigma) - x= \frac{\left(-x+\Sigma^2/2\right)^2}{2\Sigma^2}.
$$
We omit the details of the proof which can be found in~\cite{FordeJac11, FordeJacMij10}.


\end{document}